\documentclass[12pt,reqno]{amsart}
\usepackage[foot]{amsaddr}
\usepackage{amssymb}

\newtheorem{theorem}{Theorem}[section]
\newtheorem*{theorem*}{Theorem}
\newtheorem{corollary}{Corollary}[theorem]
\newtheorem{lemma}[theorem]{Lemma}
\newtheorem{proposition}[theorem]{Proposition}
\newtheorem*{remark}{Remark}

\usepackage[margin=0.9in]{geometry}

\usepackage{graphicx}
\usepackage{caption}
\usepackage{subcaption}

\usepackage[ruled]{algorithm2e}

\usepackage{enumerate}
\usepackage{xcolor}

\usepackage{hyperref}
\hypersetup{colorlinks,allcolors=blue}
\usepackage{url}

\title[Surface defects in TFW model]{Analysis of quasi-planar defects using the Thomas--Fermi--von Weizs\"acker model}
\author{Dharamveer Kumar, Amuthan A. Ramabathiran}
\address[Dharamveer Kumar]{Department of Aerospace Engineering, Indian Institute of Technology Bombay, Mumbai 400042, India.}
\email{dharamveer\_kumar@iitb.ac.in}
\address[Amuthan A. Ramabathiran]{Aerospace Engineering, California Polytechnic State University, San Luis Obispo, CA 93407.}
\email{aramabat@calpoly.edu}
\begin{document}

\begin{abstract}
We analyze the convergence of the electron density and relative energy with respect to a perfect crystal of a class of volume defects that are compactly contained along one direction while being of infinite extent along the other two using the Thomas--Fermi--von Weizs\"acker (TFW) model. We take advantage of prior work on the thermodynamic limit and stability estimates in the TFW setting, and specialize it to the case of quasi-planar defects. In particular, we prove that the relative energy of the defective crystal with respect to a perfect crystal is finite, and in fact conforms to a well-posed minimization problem. In order to show the existence of the minimization problem, we modify the TFW theory for thin films and establish convergence of the electronic fields due to the perturbation caused by the quasi-planar defect. We also show that perturbations to both the density and electrostatic potential due to the presence of the quasi-planar defect decay exponentially away from the defect, in agreement with the  known locality property of the TFW model. We use these results to infer bounds on the generalized stacking fault energy, in particular the finiteness of this energy, and discuss its implications for numerical calculations. We conclude with a brief presentation of numerical results on the (non-convex) Thomas-Fermi-von Weizs\"acker-Dirac (TFWD) model that includes the Dirac exchange in the universal functional, and discuss its implications for future work. 
\end{abstract}

\maketitle

\section{Introduction}
The analysis of point, line, and surface defects in crystalline solids occupies a central role in materials science and engineering. The structure and energetics of these defects is best understood by solving the Schr\"odinger equation for the appropriate nuclear configurations. In practice, the only viable route for studying them is through numerical computations, but a purely mathematical study of these defects is of significant theoretical interest since they can provide useful bounds for the electronic field variables and the energy, scaling laws that provide insights into the effect of system size, and/or suggest new variational principles that suggest fast algorithms tailored for specific classes of defects. The present work focuses on the behavior of the electronic fields and energy corresponding to a class of quasi-planar defects using an approximate reformulation of the Schr\"odinger equation, namely the Thomas--Fermi--von Weizs\"acker (TFW) model. 

The TFW model is a type of Density Functional Theory (DFT), where the primary variable is the electron density, unlike wavefunction based methods for the Schr\"odinger equation like the Hartree-Fock method. DFT is, in principle, an exact reformulation of the Schr\"odinger equation due to Hohenberg and Kohn \cite{HK64} when one is interested only in the ground state properties of an atomic system. The simplest, and original, formulation of DFT is the Thomas--Fermi (TF) model \cite{Thomas27, Fermi27}, but it is unsatisfactory as a quantum chemistry model since atoms do not bind according to this model; see however \cite{Solovej16} for the surprising effectiveness of the TF model for predicting the size of atoms. 

A variety of modifications have been proposed to rectify the shortcomings of the TF model. The simplest of these are the von Weizs\"acker correction for the kinetic energy functional \cite{vw35} and the Dirac correction to model exchange interactions \cite{Dirac30}. Of these, the former has the virtue of retaining the convexity of the energy functional with respect to the electron density, while the latter renders it non-convex. The mathematical analysis of these models, accordingly, is more completely understood in the former case: the mathematical analysis of the TFW model for atoms and molecules was carried out in \cite{BBL81}.

We mention at the outset that while the TFW model is rarely, if ever, used in practical computations, it is a natural starting point for mathematical analysis since it offers a reasonable trade-off between the ability to do rigorous analysis and being physically accurate. The typical point of view adopted in these studies is that despite their shortcomings, the mathematical tools developed for these analyses are expected to have wider applicability to the more physically realistic DFT models. A mathematical outlook on more realistic models, and the difficulties associated with them, can be found in \cite{Lieb83}.

A systematic mathematical analysis of the Thomas-Fermi model was carried out in the seminal work of Lieb and Simon \cite{LS77}. In addition to proving the existence and uniqueness of a ground state density for a given atomic configuration, this work also systematically established the existence of a thermodynamic limit for the TF model for an infinitely extended perfect crystal; it is to be noted that the nomenclature \emph{thermodynamic} in this context---adopted in the present work too---is meant to indicate the passage from a finite crystal to its infinite extension. A crucial ingredient of this proof is the convexity of the energy functional when expressed as a function of electron density. The corresponding thermodynamic limit analysis of the TFW model is presented in \cite{CLBL98}. A key outcome of this work is that strong qualitative results for electronic fields can be proven for TFW under very weak assumptions regarding the underlying atomic nuclei positions, without requiring them to be periodic.

Rigorous analysis of defects in the context of the TFW model has received much lesser attention. Point defects have been studied in \cite{CE11}, where it was also established that crystals with point defects are neutral in the TFW setting. Bounds on how the electron density changes when the nuclear configuration is perturbed were developed in \cite{NO17}. The exponential decay of the electron density away from a surface of a semi-crystal is proved in \cite{Blanc06}.

The present work builds on these previous efforts and applies them to the special case of quasi-planar defects. The interest in this class of defects stems from the prominent role that surface defects like generalized stacking fault surface energies, introduced by Vitek \cite{Vitek92}, play in understanding dislocation core structures in continuum models of dislocations like the Peierls-Nabarro model \cite{peierls40, nabarro47}. In practice, the generalized stacking fault surfaces are computed by considering two semi-crystals that are shifted relative to each other by a specified amount of slip; see \cite{LKBK01} for a representative calculation. The goal of this work is to specialize the aforementioned results and study their implications for quasi-planar defects like generalized stacking fault surfaces. We study a slightly modified geometric perturbation that also covers defects like grain and twin boundaries, though we do not analyze these specific defects here. 

The larger goal of this work is to develop rigorous bounds for the relative energy of quasi-planar defects, with reference to a perfect crystal, using the TFW model, and establish a variational principle for the same. To accomplish this, we focus on a tricrystal configuration where the defect is sandwiched in a slab between two perfect crystals. The rationale for this choice is that it permits a much cleaner analysis in the thermodynamic limit, while also enlarging the class of defects from surface defects to slab-like defects like grain boundaries. We use the term \emph{quasi-planar} to describe these defects for this reason. We prove using the aforementioned thermodynamic limit analyses that the electronic fields in the defective crystal converge, and that the relative energy for a class of quasi-planar defects is finite. Crucially, we show that the relative energy satisfies a well-posed minimization principle. We also present results showing the exponential decay of the electronic fields due to the defect away from the defect core, as expected from the known local nature of the TFW model. We then present an application of these results to study stacking fault energies in crystals, and, in particular, show that the stacking fault energy is bounded from above. In practice, calculations with infinitely large domains are not feasible; the usual choice is the imposition of an artificial periodic boundary condition, or a large domain with appropriate Dirichlet boundary conditions. Our analysis reveals that either of these choices leads to controlled errors in the context of the TFW model.  

Though the primary focus of this work is the TFW model, we also present simple numerical results using the TFW model with Dirac exchange---abbreviated as TFWD in the sequel---which is non-convex. The analytical methods that are developed in this work, and previous research efforts in this context, do not readily translate to this and more complex DFT models that include exchange and correlation effects since they use convexity of the energy functional in an essential way. The numerical results we provide indicate that many of the qualitative properties of the TFW response carry over to the TFWD model too, thereby hinting at the possibility of alternative approaches to the thermodynamic limit problem that do not rely on convexity. We do not, however, present such an analysis here.

The outline of the paper is as follows. Section 2 collects known mathematical results on the TFW model that are relevant for this work and sets the overall notation. The central problems that are tackled in this work, namely the convergence of density and energy for a class of quasi-planar defects, the finiteness of their relative energy with respect to a perfect crystal, and the well-posed minimization problem for the relative energy are outlined in section 3, highlighting the key theorems in this regard. Detailed proofs of the key theorems are provided in section 4. A brief discussion of the relevance of these results, including the computation of bounds for the generalized stacking fault energy and numerical results on the TFWD model are provided in section 5. Supplementary results are provided in section 6. Appendix~\ref{app:tfw_thms} collects precise statements of known theorems pertaining to the thermodynamic limit and stability estimates. The details of the numerical algorithm employed in the discussion are presented in Appendix~\ref{app:tfw_numerics}. 

\section{Mathematical Background for TFW Models}
\label{sec:mathematical_background}
We begin by providing an overview of known mathematical results pertaining to the TFW model in the context of a crystal, with and without defects, that are relevant for the ensuing discussion. Detailed statements of the relevant theorems are provided in Appendix~\ref{app:tfw_thms}. We follow the notational conventions in \cite{CLBL98}. 

Consider a finite collection of point nuclei with unit charge located at $\Lambda \subset \mathbb{Z}^3$. Let $|\Lambda|$ represent the total nuclear charge, and $\rho$ represent the electron density of $|\Lambda|$ electrons that occupy the three dimensional space around the nuclear configuration $\Lambda$. We restrict ourselves in the present work to atomic systems that have no net charge. According to the TFW model, the total energy of this atomic system is computed by evaluating the following energy functional of the density:
\begin{equation} \label{eq:energy_TFW}
\begin{split}
I_\Lambda(\rho) &= E_\Lambda(\rho) + \frac{1}{2}\sum_{k \neq l \in \Lambda} \frac{1}{\lVert k - l\rVert},\\
E_\Lambda(\rho) &= C_W\int_{\mathbb{R}^3} \lVert \nabla \sqrt{\rho}(x) \rVert^2\,dx + C_{TF}\int_{\mathbb{R}^3} \rho^{\frac{5}{3}}(x)\,dx\\
 & - \int_{\mathbb{R}^3} \left(\sum_{k \in \Lambda} \frac{1}{\lVert x - k\rVert}\right)\rho(x)\,dx + \frac{1}{2}\int_{\mathbb{R}^3}\int_{\mathbb{R}^3} \frac{\rho(x)\rho(y)}{\lVert x - y\rVert} \, dx \,dy. 
\end{split}
\end{equation}
Here $E_\Lambda(\rho)$ is the energy associated with the electron-electron and electron-nuclear interactions. All electrostatic interactions are modeled using the Coulomb potential $V(x) = 1/\lVert x \rVert$. In comparison to a classical electrostatic formulation, the additional terms are the first two in the expression for $E_\Lambda(\rho)$ expressing the kinetic energy of the electrons. The second of these is identical to that in the TF model, while the first is a correction due to von Weizs\"acker. It is to be noted that on top of the approximation that is incurred by this expression for the kinetic energy, a crucial omission in the TFW model is the absence of any terms to handle exchange and correlation effects for the electron-electron interactions. For the purposes of the mathematical analysis that follows, the constants $C_{TF}$ and $C_W$ are set to unity; the corresponding modification to the results when these constants are not equal to $1$ is straightforward. It was shown in \cite{BBL81}, among other things, that a unique minimizing density exists for the variational TFW problem 
\begin{equation} \label{eq:TFW_min}
\text{inf}\;\left\{I_\Lambda(\rho) \;:\; \sqrt{\rho} \in H^1(\mathbb{R}^3), \rho \ge 0, \int_{\mathbb{R}^3} \rho = |\Lambda|\right\}.
\end{equation}
We present next the existence and uniqueness theorems for a special case when $\Lambda$ is infinitely large.

\subsection{Thermodynamic Limits}
Consider the case $\Lambda = \mathbb{Z}^3$. The energy of the crystal is infinite in this case, but we can meaningfully define the energy per unit cell as the limit $E_\Lambda(\rho_\Lambda)/|\Lambda|$ as $\Lambda$ fills up $\mathbb{Z}^3$ in the sense of van Hove sequences; see \cite{CLBL98} for details. The study of this limit, called the \emph{thermodynamic limit}, is carried out in the seminal work of \cite{CLBL98} where it was shown that the energy per unit cell converges to a finite value, and that the electron density converges to a periodic density with the periodicity of the lattice. More precisely, they show that the limit density and energy per unit cell are obtained by solving the following periodic minimization problem over the unit cell $\Gamma_0 = (-1/2, 1/2]^3 \subset \mathbb{R}^3$:
\begin{equation} \label{eq:periodic_TFW_thermo}
\begin{split}
\inf \; \biggl\{E_{\text{per}}(\rho) \biggr. &\;:\; \left. \sqrt{\rho} \in H^1_{\text{per}}(\mathbb{R}^3), \rho \ge 0, \int_{\Gamma_0} \rho = 1\right\}, \\
E_{\text{per}}(\rho) &= \int_{\Gamma_0} \lVert \nabla \sqrt{\rho}(x) \rVert^2\,dx + \int_{\Gamma_0} \rho^{\frac{5}{3}}(x)\,dx\\
 & - \int_{\Gamma_0} \rho(x)G(x)\,dx + \frac{1}{2}\int_{\Gamma_0}\int_{\Gamma_0} \rho(x)\rho(y)G(x - y) \, dx \,dy.
\end{split}
\end{equation}
In the periodic minimization problem \eqref{eq:periodic_TFW_thermo}, $G$ is the Green's function obtained by solving the following Partial Differential Equation (PDE):
\begin{equation} \label{eq:green_TFW}
-\nabla^2 G(x) = 4\pi\left(-1 + \sum_{k \in \mathbb{Z}^3} \delta(x - k)\right), \quad \int_{\Gamma_0} G(x)\,dx = 0.
\end{equation}
It is easily verified that solution of the PDE \eqref{eq:green_TFW} is given by
\begin{equation} \label{eq:green_TFW_soln}
G(x) = 4\pi\sum_{k \in (2\pi\mathbb{Z})^3 \setminus \{0\}} \frac{\exp (ik\cdot x)}{\lVert k \rVert^2} + C,
\end{equation}
where $C$ is a constant chosen to satisfy the integral constraint in \eqref{eq:green_TFW}. 

It is convenient to smear out the nuclear charge to a nuclear density satisfying certain additional physically relevant conditions. To this end, we follow \cite{CLBL98} and \cite{NO17} to introduce the following class $\mathcal{M}$ of smeared out nuclear densities:
\begin{equation} \label{eq:n_smeared}
\begin{split}
\mathcal{M} &= \{n \in L^2_{\text{unif}}({\mathbb{R}^3}) \;:\; n \ge 0, (H_1), (H_2) \text{ hold}\},\\
L^2_{\text{unif}}(\mathbb{R}^3) &= \{f \in L^2_{\text{\text{loc}}}(\mathbb{R}^3) \;:\; \sup_{x \in \mathbb{R}^3} \lVert f \rVert_{L^2(B_1(x))} < \infty\},\\
(H_1) &:\; \sup_{x \in \mathbb{R}^3} \int_{B_1(x)} n(y)\,dy < \infty,\\
(H_2) &:\; \lim_{R \to \infty} \inf_{x \in \mathbb{R}^3} \frac{1}{R}\int_{B_R(x)} n(y)\,dy = +\infty.
\end{split}   
\end{equation}
Conditions $(H_1)$ and $(H_2)$ restrict $\mathcal{M}$ to nuclear densities that do not have infinitely dense clusters and regions with large voids, respectively. Note that the nuclear density being a $L^2_{\text{unif}}(\mathbb{R}^3)$ function already implies $(H_1)$. For nuclear densities $n \in \mathcal{M}$ which are periodic over $\Gamma_0$ it is shown in \cite{CLBL98} that the energy functional \eqref{eq:periodic_TFW_thermo} can be written as the following functional of $u := \sqrt{\rho}$:
\begin{equation} \label{eq:energy_TFW_smeared}
\begin{split} 
I_n(u) &= \int_{\Gamma_0} \lVert \nabla u(x) \rVert^2\,dx + \int_{\Gamma_0} u^{10/3}(x)\,dx + D_{\Gamma_0}(n - u^2, n - u^2),\\
D_{\Gamma_0}(f, g) &= \iint_{{\Gamma_0} \times {\Gamma_0}} f(x)g(y)G(x-y)\,dx dy, \quad f, g \in L^{2}({\Gamma_0}).
\end{split}
\end{equation}
We note for subsequent use that the Euler-Lagrange equations for the minimization problem associated with the functional \eqref{eq:energy_TFW_smeared} is given by the following Schr\"odinger-Poisson system:
\begin{equation} \label{eq:TFW_smeared_EL}
\begin{split}
-\nabla^2 u &+ \frac{5}{3}u^{7/3} + \phi u = 0,\\
u \geq 0, \\
-\nabla^2 \phi &= 4\pi(u^2 - n).
\end{split}
\end{equation}
Note that this system is well-defined even when $n$ is not periodic; please refer to theorem~\ref{thm:TFW_smeared_EL_exst_uniq} for more details.

The foregoing kind of analysis has also been extended for other semi-infinite configurations like polymer and thin-films in ~\cite{BTF00}. The precise mathematical statements are listed in Appendix~\ref{app:tfw_thms}.  

\subsection{Comparison theorems}
For $M > 0$ define
\begin{equation}
  \mathcal{M}^M = \{ m \in \mathcal{M} : \|m\|_{L^2_{unif}(\mathbb{R}^3)} \leq M\}. 
\end{equation}
Since the goal of the present work is an analysis of surface defects that involves comparing a defective configuration with a reference defect-free configuration, the following general comparison theorems provide useful stability estimates for solutions of the TFW model:

\begin{theorem}\label{thm:ortner_TFW_stability}
\textit{\textup{[}Theorem 3.1 in \cite{NO17}\textup{]}}
Let $m_1, m_2 \in \mathcal{M}^M$ be two given nuclear densities, and let $(u_1, \phi_1)$ and $(u_2, \phi_2)$ be the corresponding unique ground state solutions, respectively. Then for any $y \in \mathbb{R}^3$ there exist positive constants $C_1$ and $C_2$ such that  
\begin{equation}
\label{eq:TFW_stability}
\sum_{|s| \leq 2}|\partial_s(u_1 - u_2)(y)|^2 + |(\phi_1 - \phi_2)(y)|^2 \leq C_1 \int_{\mathbb{R}^3} |(m_1-m_2)(x)|^2 e^{-C_2\lVert x-y \rVert}\,dx.
\end{equation}
The constants $C_1, C_2$ depend only on the uniform bounds of $m_1$ and $m_2$.  
\end{theorem}

\begin{proposition}\label{prop:ortner_TFW_stability_self}
\textit{\textup{[}Proposition 4.1 in \cite{NO17}\textup{]}} Let $m_1, m_2 \in \mathcal{M}^M$ be two given nuclear densities, and let $(u_1, \phi_1)$ and $(u_2, \phi_2)$ be the corresponding unique ground state solutions, respectively.
Further, let $\Omega \subset \mathbb{R}^3$ be open such that $m_2 = m_1$ on $\Omega$.
Then for all $y \in \Omega$
\begin{equation}
\sum_{|s| \leq 2}|\partial_s(u_2 - u_1)(y)|^2 + |(\phi_2 - \phi_1)(y)|^2 \leq C_1e^{-C_2\text{dist}(y, \partial \Omega)}.
\end{equation}
The constants $C_1, C_2$ are positive, independent of $\Omega$, and depend only on the uniform bounds of $m_1$. 
\end{proposition}

What these theorems establish is a \emph{locality} property for TFW models which states that perturbations in the electronic fields are controlled by perturbations in the nuclear distribution. In this work, we use these estimates to show the existence of a variational principle  corresponding perturbation of energy for a special class of perturbations.

\section{Problem Setting and Main Results}
We study the relative energies of two nuclear densities $m_1, m_2 \in \mathcal{M}^M$ such that the nuclear perturbation $\nu := m_2 - m_1$ is \emph{spatially confined} in a manner that will be made precise later. We are interested in understanding how the difference in energy $I_{m_2} - I_{m_1}$ is related to the perturbation $\nu$; the energy $I_m$ corresponding to a nuclear distribution $m$ is given as in equation~\eqref{eq:energy_TFW_smeared}, with the integrals defined over the appropriate domains, and the Green's function defined correspondingly. We analyze a special class of perturbations compactly contained in semi-infinite domains where the individual energies $I_{m_1}$ and $I_{m_2}$ are infinitely large and demonstrate that the difference in energy $I_{m_2} - I_{m_1}$ is finite, and further satisfies a variational principle. Throughout our analysis, we assume that both the systems corresponding to nuclear distributions $m_1$ and $m_2$ are charge neutral. 

To be more precise, we study a class of extended and quasi-planar defects in a perfect crystalline lattice that are contained within a strip $\mathbb{R} \times \mathbb{R} \times [-L_0/2, L_0/2] \subseteq \mathbb{R}^3$, where $L_0$ is an odd natural number, and periodic along the first two dimensions with periodicity 1, as shown in Figure~
\ref{fig:planar_defect_schematic}. We note that the analysis that follows can be readily extended to situations with unequal periodicities along the first two dimensions and $L_0$ not being an odd number.  Important examples of defects of this type in materials include grain boundaries, lenticular twins, and a pair of stacking faults resulting in a tricrystal configuration.  

\begin{figure}[h]
\centering
\includegraphics[width=0.35\textwidth]{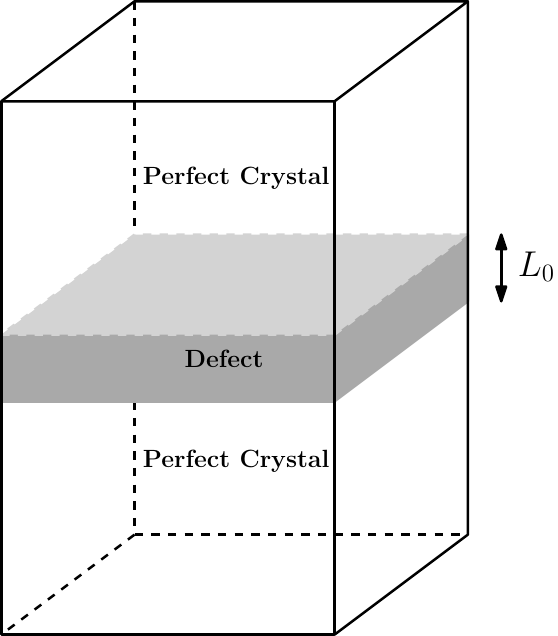}
\caption{Typical geometry of quasi-planar defects studied in this work. The defective region can be thought of as a thin film that is sandwiched between two perfect crystals of the same material and identical orientation. Periodicity is assumed along the in-plane directions of the defect. Note that we do \emph{not} assume that the nuclear distribution in the defective region is homogeneous.}
\label{fig:planar_defect_schematic}
\end{figure}

For the thermodynamic limit analysis, we consider supercells of the form $\Gamma_L := [-\frac{1}{2},\frac{1}{2}) \times [-\frac{1}{2},\frac{1}{2})  \times [-\frac{L}{2},\frac{L}{2})$ for some $L \in \mathbb{N}$ such that $L$ is odd and $L > L_0$; we will denote by $\Gamma_\infty$ the unit cell in the limit $L \to \infty$, just as in section~\ref{sec:thinfilm}. Let $\mathcal{M}_{xy}$ denote the set of all densities in $\mathcal{M}^M$ that are also periodic along the $x$ and $y$ directions with unit period. Given $m \in \mathcal{M}_{xy}$, we construct a periodic nuclear density $m_L \in \mathcal{M}_{xy} \cap L^2_{\text{per}}(\Gamma_L)$ as follows:
\begin{equation} \label{eq:periodic_nuclear_density}
m_L(x) = \sum_{y \in \mathcal{R}_L} (\chi_{\Gamma_L}m)(x - y).
\end{equation}
Here $\chi_A$ stands for the characteristic function of the set $A \subseteq \mathbb{R}^3$, and $\mathcal{R}_L$ denotes the lattice $\mathbb{Z} \times \mathbb{Z} \times L\mathbb{Z}$. It is easily verified that $m_L$ converges uniformly to $m$ on any compact set in the limit $L \to \infty$. Note that $m_L\vert_{\Gamma_L}=m\vert_{\Gamma_L}$. A visual representation of this procedure to obtain a periodic supercell is shown in Figure~\ref{fig:supercell_construction}.

% \begin{figure}[h]
% \begin{subfigure}{0.48\textwidth}
% \includegraphics[width=0.9\textwidth]{figures/supercell_method_left.png}
% \label{fig:m}
% \caption{$m$}
% \end{subfigure}
% ~
% \begin{subfigure}{0.515\textwidth}
% \includegraphics[width=0.91\textwidth]{figures/supercell_method_right.png}
% \label{fig:mL}
% \caption{$m_L$}
% \end{subfigure}
% \caption{{\color{red}Procedure to generate a supercell with periodic nuclear density $m_L$ from a given nuclear density $m \in \mathcal{M}_{xy}$. The original system is sliced into boxes with dimension $L$ along the out-of-plane direction of the planar defect and periodically tiled over $\mathbb{R}^3$.}}
% \label{fig:supercell_construction}
% \end{figure} 

\begin{figure}[h]
\begin{subfigure}{0.4\textwidth}
\includegraphics[width=0.9\textwidth]{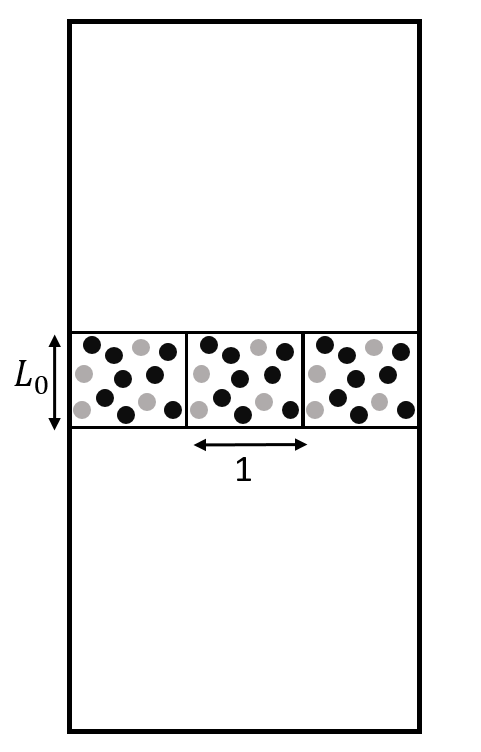}
\label{fig:m}
\caption{$m$}
\end{subfigure}
~
\begin{subfigure}{0.4\textwidth}
\includegraphics[width=0.9\textwidth]{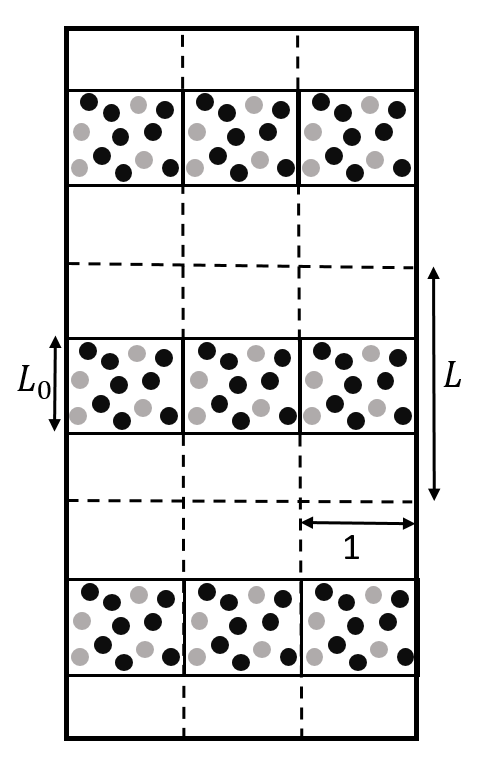}
\label{fig:mL}
\caption{$m_L$}
\end{subfigure}
\caption{Procedure to generate a supercell with periodic nuclear density $m_L$ from a given nuclear density $m \in \mathcal{M}_{xy}$. The original system is sliced into boxes with dimension $L$ along the out-of-plane direction of the planar defect and periodically tiled over $\mathbb{R}^3$.}
\label{fig:supercell_construction}
\end{figure}

Consider now a perfect cubic crystal with unit periodicity along each direction, and with nuclear charge distribution $m_1 \in \mathcal{M}^M$. Clearly, $m_1 \in \mathcal{M}_{xy}$ and $m_{1,L}=m_1$. The corresponding electron density $\rho_1 = u_1^2$, which is strictly positive, continuous, and  periodic, and potential $\phi_1$, do not depend on $L$ due to the uniqueness of the ground state electron density; see Theorem \ref{thm:TFW_smeared_EL_exst_uniq} and Theorem~\ref{thm:TFW_supercell_exst_uniq}. Without loss of generality we assume that $u_1 > 0$ and $\int_{\Gamma_1}m_1=1$; charge neutrality and periodicity then implies that $\int_{\Gamma_L} u_1^2 = \int_{\Gamma_L} m_1 = L$. Note that using~\ref{thm:TFW_supercell_exst_uniq}, $u_1 \in L^\infty(\mathbb{R}^3)$. Further assume that $m_1$ is smooth, which implies $\phi_1 \in L^\infty(\mathbb{R}^3)$ \cite{Blanc06}.

Within this context, we consider special perturbations of the perfect crystal that constitute a class of quasi-planar defects characterized by perturbed nuclear densities, $m_2$, of the form $m_2 = m_1 + \nu$, where $\nu$ belongs to the following set:
\begin{equation} \label{eq:nearly_planar_perturbation}
\mathcal{P} = \{ \nu \in \mathcal{M}_{xy} \;:\; m_1 + \nu \ge 0, \;  \exists\; L_0 \in \mathbb{N}, L_0 \text{ odd, }\text{ s.t. } \text{supp}(\nu\vert_{\Gamma_\infty})  \subset \overline{\Gamma}_{L_0}\}.
\end{equation}
The final condition above restricts the perturbation $\nu$ to be confined to a finite region of width $L_0$ along the $z$ direction. Note that $m_2 \in \mathcal{M}_{xy}$ and $\int_{\Gamma_\infty}\nu$ is a constant that depends on the difference in number of atoms between the original and perturbed systems. 

For such perturbations, we consider successively a series of nuclear densities $m_{2,L}$---constructed according to the procedure shown in equation~\eqref{eq:periodic_nuclear_density} and Figure~\ref{fig:supercell_construction}---for increasingly large values of $L$, and study the corresponding periodic supercell problem. In particular, we demonstrate the following:
\begin{enumerate}[(i)]
\item The ground state solution $(u_{2,L}, \phi_{2,L})$ of the supercell with nuclear density $m_{2,L}$ converges in the limit $L \to \infty$ to the pair $(u_2,\phi_2)$ that solves the following Schr\"odinger-Poisson system on $\mathbb{R}^3$:
\begin{equation} \label{eq:TFW_converged_supercell_pde}
\begin{split}
-\nabla^2 u_2 &+ \frac{5}{3}u_2^{7/3} + \phi_2 u_2 = 0,\\
-\nabla^2 \phi_2 &= 4\pi(u_2^2 - m_2).
\end{split}
\end{equation}
 As a direct consequence of proposition~\ref{prop:ortner_TFW_stability_self}, we show that $u_2$ exponentially decays to $u_1$ outside the perturbation core $\Gamma_{L_0}$, thereby establishing the locality of the electron density for the considered class of nuclear perturbations. Proofs of these statements are covered in Lemma~\ref{thm:denconvg} and Lemma~\ref{thm:density_decay}.
\item Denoting by $I^L_{m_1}$ and $I^L_{m_1 + \nu}$ the energy per unit supercell for the nuclear distributions $m_{1,L}$ and $m_{2,L}$, respectively, we show that the difference in energy between the perturbed system and the reference perfect crystal, namely
\begin{equation} \label{eq:gamma_L_defn}
\gamma^L_1(\nu) = I^L_{m_1 + \nu} - I^L_{m_1},
\end{equation}
converges in the limit $L \to \infty$ to a finite value, as a consequence of the locality of the electronic fields. It is worth emphasizing that both $I^L_{m_1}$ and $I^L_{m_1 + \nu}$ are infinitely large in the limit $L \to \infty$. A proof of this statement is provided in Theorem~\ref{thm:convggam}.
\item Finally, we establish a variational principle for the relative energy introduced above. In particular, we demonstrate that there exists a functional $\mathcal{E}_{m_1}^{\nu}$ such that 
\begin{equation} \label{eq:gamma_L_conv}
\lim_{L \to \infty} \gamma^L_1(\nu) = \inf_{w}\mathcal{E}^\nu_1(w).
\end{equation}
More precisely, we prove the following result:
\begin{theorem*}
The electron density $u_2^2$ of the perturbed problem is such that $u_2 - u_1$ minimizes the following energy functional:
\begin{equation*}
\lim_{L \to \infty} \gamma^L_1(\nu) := \gamma^\nu_1 = \inf_{w \in \mathcal{Q}_1^{\nu}} \; \mathcal{E}^\nu_1(w) = \mathcal{E}^\nu_1(u_2 - u_1),
\end{equation*}
where 
\begin{displaymath}
\begin{split}
\mathcal{E}^\nu_1(w) &= \int_{\Gamma_\infty} \lVert\nabla w\rVert^2 + \int_{\Gamma_\infty} (|u_1 + w|^{10/3}-u_{1}^{10/3} - \frac{10}{3}u_1^{7/3}w)\\ 
 & + \int_{\Gamma_\infty} \phi_{1} (w^2 - \nu) + \frac{1}{2}D_G((u_1 + w)^2 - u_1^2 -\nu, (u_1 + w)^2 - u_1^2 -\nu),\\
\mathcal{Q}_1^{\nu} &= \{w \in H^1_{per}(\Gamma_\infty) \;:\; u_1 + w \geq 0, |x_3|u_1w \in L^1(\Gamma_\infty), |x_3|w^2 \in L^1(\Gamma_\infty),\\
 &\qquad\qquad\qquad\qquad\int_{\Gamma_\infty} (u_1 + w)^2 - u_1^2 = \int_{\Gamma_\infty}\nu \}.
\end{split}
\end{displaymath}    
\end{theorem*}
A proof of this theorem is presented in the next section as Theorem~\ref{thm:reflexivemin}.
\end{enumerate}

The second result listed above justifies the definition of a \emph{surface energy} for a given perturbation $\nu \in \mathcal{P}$; we clarify this with an application of these results to the case of a stacking fault in a perfect crystal. The third result listed above is a key contribution of this work, and establishes a variational principle for the relative energy.

\section{Thermodynamic limit analysis of electron density and surface energy}
This section presents a detailed analysis of the key results outlined in the previous section. 

\subsection{Convergence of electron density}
The convergence of the electron density of the supercell for the perturbed nuclear configuration follows as an easy consequence of the results summarized in section~\ref{sec:mathematical_background}. We also establish, following section 2 of ~\cite{NO17}, the rate of convergence and the decay property of the electron density away from the nuclear perturbation.  

Let $m_2 \in \mathcal{M}$ such that it is periodic with unit period in both the 
$x$ and $y$ directions. Then using theorem~\ref{thm:TFW_smeared_EL_exst_uniq} we obtain that there exists $(u_2, \phi_2)$ which unique solution in $L^{7/3}_{\text{loc}}\cap L^2_{\text{unif}}(\mathbb{R}^3) \times L^1_{\text{unif}}(\mathbb{R}^3)$ of the following system of PDEs:
\begin{displaymath}
    \begin{split}
        &-\nabla^2 u_2 + \frac{5}{3}u_2^{7/3} + \phi_2 u_2 = 0, \\
        & u_2 \geq 0,\\
&-\nabla^2 \phi_2 = 4\pi(u_2^2 - m_2). \\
    \end{split}
\end{displaymath}
Further, $\inf_{\mathbb{R}^3} u > 0$, $u \in L^\infty(\mathbb{R}^3) \cap H^2_{loc}(\mathbb{R}^3)$ and both $u_2$ and $\phi_2$ are periodic on $\Gamma_\infty$.

We then obtain periodic density $m_{2,L} \in L^2_{\text{per}}(\Gamma_L)$ using the procedure specified in equation~\eqref{eq:periodic_nuclear_density}. Theorem~\ref{thm:TFW_supercell_exst_uniq} then guarantees the existence of a unique continuous ground state density $\rho_{2,L} = u_{2,L}^2 > 0$ and potential $\phi_{2,L}$ for every $L$. Without loss of generality, we assume $u_{2,L} >0$. We now prove the first result listed in the previous section, namely that the electron density $u_{2,L}^2$ and potential $\phi_{2,L}$ converge in the limit $L \to \infty$ to the unique solution $u_2^2$ and $\phi_2$, respectively, of a Schr\"odinger-Poisson system of PDEs.

\begin{lemma}
\label{thm:denconvg} Let $\nu = m_2-m_1$ be such that $\nu|_{\Gamma_\infty}$ is compactly supported in $\Gamma_\infty$ and belongs to $L^2(\Gamma_\infty)$. Further, let $v_L = u_{2,L} - u_1$ and $\phi_L = \phi_{2,L} - \phi_1$ be such that $\|v_L\|_{H^2(\Gamma_L)}$ and $\|\phi_L\|_{L^2(\Gamma_L)}$ are bounded independently of $L$. Then  
$u_{2,L}$ converges to $u_2$ in $H^1_{loc}(\mathbb{R}^3)$ and in $L^p_{loc}(\mathbb{R}^3)$ for $1 \leq p < 6$, and $\phi_{2,L}$ converges to $\phi_2$ in $L^p_{loc}(\mathbb{R}^3)$ for $1 \leq p < 6$.  
\end{lemma}
\begin{proof}
 From theorem~\ref{thm:TFW_supercell_exst_uniq},
\begin{equation} 
-\nabla^2 \phi_L = 4\pi (v_L^2 +2u_1v_L -\nu).
\end{equation}
That implies that,
$$ \int_{\Gamma_L} \|\nabla \phi_L\|^2 = 4\pi \int_{\Gamma_L}(v_L^2 +2u_1v_L -\nu)\phi_L.$$ 
Therefore, $\|\phi_L\|_{H^1(\Gamma_L)}$ is bounded independently of $L$, since $\|\phi_L\|_{L^2(\Gamma_L)}$ is bounded independent of $L$. Using lemma~\ref{lem:denconvg_help} (which follows Lemma 4.6 in \cite{CE11}) we get that $v_L$ and $\phi_L$, along a subsequence, convergence to some $v \in H^1_{loc}(\mathbb{R}^3)$ and $\phi \in H^1_{loc}(\mathbb{R}^3)$, respectively, weakly in $H^1_{loc}(\mathbb{R}^3)$, strongly in $L^p_{loc}(\mathbb{R}^3)$ for $1 \leq p < 6$ and almost everywhere in $\mathbb{R}^3$.

Equivalently, $u_{2,L}$ converges to $u_1 + v$ and $\phi_{2,L}$ converges to $\phi_1 + \phi$. Since $u_{2,L}$ is positive for any $L$, that implies $u_1+v$ is non-negative. Taking the limit, in the sense of distributions, of equation~\eqref{eq:TFW_smeared_EL} corresponding to $m_{2,L}$ we infer that $u_1 +v$ and $\phi_1 + \phi$ satisfy
\begin{displaymath}
    \begin{split}
        &-\nabla^2 (u_1 + v) + \frac{5}{3}(u_1 + v)^{7/3} + (\phi_1 + \phi) (u_1 + v) = 0, \\
        & (u_1+v) \geq 0, \\
&-\nabla^2 (\phi_1 + \phi) = 4\pi((u_1 + v)^2 - m_2).
    \end{split}
\end{displaymath}

Using theorem~\ref{thm:TFW_smeared_EL_exst_uniq}, the above set of equations has a unique solution which is already satisfied by $(u_2, \phi_2)$. Hence 
$u_2 = u_1+v$ and $\phi_2=\phi_1 + \phi$. We can use the uniqueness of the solution to remove the subsequence condition for convergence in $L^p_{loc}(\mathbb{R}^3)$ for $1 \leq p < 6$. 

Finally, let $\nabla v_L$ converge to some $g \in H^1_{loc}(\mathbb{R}^3; \mathbb{R}^3 )$ in $L^2_{loc}(\mathbb{R}^3; \mathbb{R}^3)$,  using lemma~\ref{lem:denconvg_help} and the fact that $\|v_L\|_{H^2(\mathbb{R}^3)} $ is bounded independently of $L$. Since $v_L$ converge to $v$ weakly in $H^1_{loc}(\mathbb{R}^3)$ we get $g = \nabla v$. Hence $v_L$ converges to $v$ in $H^1_{loc}(\mathbb{R}^3)$.
\end{proof}

We showed in Lemma~\ref{thm:denconvg} that both $v_L$ and $\phi_L$, when they are uniformly bounded in some sense,  converge in $L^p_{loc}(\mathbb{R}^3)$ for $1 \leq p < 6$ to bounded functions $v = u_2 - u_1$ and $\phi = \phi_2 -\phi_1$, respectively. Now let \( \nu \in \mathcal{P} \) is a prescribed quasi-planar nuclear perturbation as defined in equation~\eqref{eq:nearly_planar_perturbation}. Then  \( m_2 = m_1 + \nu \) implies $m_2 \in \mathcal{M}_{xy}$. The following result establishes the decay properties of $v$ and $\phi$ away from the core of the nuclear perturbation $\Gamma_{L_0}$, and also shows the rate of convergence of $v_L$ and $\phi_L$ to $v$ and $\phi$, respectively.

\begin{lemma}
\label{thm:density_decay}
 For every $y \in \Gamma_\infty \setminus \Gamma_{L_0}$, it holds that $$|\partial_s v(y)| + |\phi(y)| \leq k_1e^{-k_2 dist(y, \Gamma_{L_0})},$$ for positive constants $k_1, k_2$ that depend only on the uniform bounds on $m_1$, and $s$ is a multi-index such that $|s| \in \{0, 1, 2\}$.
In addition, for $L > L_0$, $$\|u_{2,L}-u_2\|_{W^{k,p}(\Gamma_L)} + \|\phi_{2,L}-\phi_2\|_{L^p(\Gamma_L)} \leq k_3e^{-k_4(L-L_0)}$$ for $k\in \{0, 1, 2\}$, $1 \leq p \leq \infty$, and positive constants $k_3, k_4$ that depend only on uniform bounds of $m_2$ and $L_0$. 
\end{lemma}
\begin{proof}
Define $\Omega_{L} = (-\infty, \infty) \times (-\infty, \infty) \times (-L/2, L/2)$ for all $L$. Since $m_2$ is identical to $m_1$ on the open set $\mathbb{R}^3 \setminus \overline{\Omega}_{L_0}$, we have using proposition~\ref{prop:ortner_TFW_stability_self} that
$$\sum_{|s|\leq 2}|\partial_s (u_2 - u_1)(y)|^2 + |(\phi_2-\phi_1)(y)|^2 \leq \frac{k_1^2}{4}e^{-2k_2 dist(y, \Omega_{L_0})},$$ 
for every $y \in \Gamma_\infty \setminus \Gamma_{L_0}$ and the positive constants $k_1, k_2$ depend only on the uniform bounds on $m_1$. Each term on the left hand side is bounded by $\frac{k_1}{2}e^{-k_2 dist(y, \Omega_{L_0})}$. Hence, using $v=u_2-u_1$, $\phi=\phi_2-\phi_1$ we get, 
$$|\partial_s v(y)| + |\phi(y)| \leq k_1e^{-k_2 dist(y, \Gamma_{L_0})},$$
for every $y \in \mathbb{R}^3 \setminus \overline{\Omega}_{L_0}$ and every multi-index $s$ such that $|s| \in \{0, 1, 2\}$. Moreover, integrating the above equation over $\Gamma_\infty \setminus \Gamma_{L_0}$ and using the fact that $(v,\phi) \in L^\infty(\mathbb{R}^3) \times L^\infty(\mathbb{R}^3)$ we see that $v \in W^{k,1}(\Gamma_\infty)$ and $\phi \in L^1(\Gamma_\infty)$ for $k \in \{0, 1, 2\}$. Further, using $(v,\phi) \in L^\infty(\mathbb{R}^3) \times L^\infty(\mathbb{R}^3)$ and are periodic over $\Gamma_\infty$ we conclude that $v \in W^{k,p}_{per}(\Gamma_\infty)$ and $\phi \in L^p_{per}(\Gamma_\infty)$ for $k \in \{0, 1, 2\}$ and $1 \leq p \leq \infty$.

To prove the second inequality, we note that $m_{2,L} \in \mathcal{M}$, and for any $L > L_0$, $m_{2,L}$ is identical to $m_2$ in the interior of $\Omega_{2L-L_0}$. We therefore have using proposition~\ref{prop:ortner_TFW_stability_self} that
\begin{equation}
\label{eq:dens_decay}
    \sum_{|s|\leq 2}|\partial_s (u_{2,L} - u_2)(y)|^2 + |(\phi_{2,L}-\phi_2)(y)|^2 \leq c_1e^{-c_2 dist(y, \partial \Omega_{2L-L_0})},
\end{equation} 
for every $y $ in the interior of $\Gamma_{L}$ and the positive constants $c_1, c_2$ depend only on the uniform bounds on $m_2$.  Integrating over $\Gamma_L$, we get
\begin{equation*}
\|u_{2,L}-u_2\|_{W^{k,p}(\Gamma_L)} + \|\phi_{2,L}-\phi_2\|_{L^p(\Gamma_L)} \leq k_3e^{-k_4(L-L_0)}
\end{equation*} 

for $k\in \{0, 1, 2\}$, $1 \leq p \leq \infty$, and positive constants $k_3, k_4$ that depend only on uniform bounds of $m_2$ and $L_0$.
\end{proof}

Lemma~\ref{thm:density_decay} provides the necessary uniform boundedness required for Lemma~\ref{thm:denconvg}. In addition, it directly implies the local convergence of the electronic density and potential. However, it requires that \( m_2 \in  \mathcal{M}^M\), whereas Lemma~\ref{thm:denconvg} does not require this condition and is therefore more general, allowing for \( m_2 \in \mathcal{M} \). This broader applicability also makes it suitable for extension to other types of defects. In the following corollary, we further extend the results of Lemma~\ref{thm:density_decay}, which will be instrumental in establishing energy convergence.

\begin{corollary}
\label{cor:supercell_field_convergence}
As $L \to \infty$,
\begin{align*}
    \chi_{\Gamma_L}v_L &\rightarrow \chi_{\Gamma_\infty} v, \\
    \chi_{\Gamma_L}\nabla v_L &\rightarrow \chi_{\Gamma_\infty} \nabla v, \\
    \chi_{\Gamma_L}\phi_L &\rightarrow \chi_{\Gamma_\infty} \phi,
\end{align*}
in $L^p(\Gamma_\infty)$ for all $1 \leq p < \infty$.
\end{corollary}

\begin{proof}
From Lemma~\ref{thm:density_decay} we have that $v_L, \|\nabla v_L\|, $ and $\phi_L $ are uniformly bounded irrespective of $L$. Moreover, $v, \|\nabla v\|,$ and $\phi$ belong to $L^\infty(\mathbb{R}^3)$.
Hence, it suffices to show convergence of the sequences mentioned in the corollary in $L^1(\Gamma_\infty)$ to establish convergence in $L^p(\Gamma_\infty)$, for all $1\leq p < \infty$, using H\"older's inequality.
 Using lemma~\ref{thm:density_decay} we have that
    \begin{displaymath}
        \begin{split}
        \|\chi_{\Gamma_L}v_L-v\|_{L^1(\Gamma_\infty)} &\leq \|\chi_{\Gamma_L}u_{2,L}-\chi_{\Gamma_L}u_2\|_{L^1(\Gamma_\infty)} + \|\chi_{\Gamma_L}v - v\|_{L^1(\Gamma_\infty)}  \\
        & = \|u_{2,L}-u_2\|_{L^1(\Gamma_L)}  + \|v\|_{L^1(\Gamma_\infty \setminus \Gamma_L)} \\
        &\leq k_3e^{-k_4(L-L_0)}  + \int_{\Gamma_\infty \setminus \Gamma_L}k_1e^{-k_2 dist(y, \Gamma_{L_0})} , 
        \end{split}
    \end{displaymath}
    where the positive constants $k_1, k_2$ depend only on the uniform bounds of $m_1$ and $k_3, k_4$ are positive constants that depend only on $L_0$ and the uniform bounds of $m_2$. Hence,
    $$ \lim_{L \rightarrow \infty}\|\chi_{\Gamma_L}v_L - v\|_{L^1(\Gamma_\infty)} = 0 .$$ 
    The convergence of the other sequences is proved similarly. 
\end{proof}

Corollary~\ref{cor:supercell_field_convergence} also shows neutrality of charge in $\Gamma_\infty$. To see this, note that for large enough $L$, $$\int_{\Gamma_L}(v_L^2 + 2u_1v_L) = \int_{\Gamma_\infty}\nu. $$ 

Taking the limit $L \to \infty$ gives $\int_{\Gamma_\infty}(v^2 +2u_1v)= \int_{\Gamma_\infty}\nu$, which is the condition for charge neutrality, namely $\int_{\Gamma_\infty} u_2^2-u_1^2 = \int_{\Gamma_\infty} \nu $.

\begin{corollary}
\label{cor:establish_green_function}
Let $$ G(x) = -2\pi|x_3| + \sum_{k \in \mathbb{Z}^2 \times \{0\}}\left( 
\frac{1}{\lVert x-k\rVert} -\int_{(-\frac{1}{2}, \frac{1}{2})^2 \times \{0\}}\frac{dy}{\lVert x-k-y\rVert}\right).$$
Then $\phi = \phi_2 - \phi_1 = G \star_{\Gamma_\infty} (u_2^2 - u_1^2 - \nu)$ and $\phi \in H^1(\Gamma_\infty)$.     
\end{corollary}
\begin{proof}
We have from theorem~\ref{thm:TFW_supercell_exst_uniq} that 
\begin{equation} \label{eq:phiL_poisson}
-\nabla^2 \phi_L = 4\pi (v_L^2 +2u_1v_L -\nu).
\end{equation}
Taking the limit $L \to \infty$ in equation~\eqref{eq:phiL_poisson}, we get
\begin{equation}
    \label{eq:perturbed_poisson}
    -\nabla^2 \phi = 4\pi(u_2^2 - u_1^2 -\nu),
\end{equation}
It is straightforward to verify that
$$ \phi(x) = \int_{\Gamma_\infty} G(x - y)(u_2^2(y) - u_1^2(y) -\nu(y))dy$$
with $G$ defined as above is well-defined and satisfies equation~\eqref{eq:perturbed_poisson}.

To show uniqueness of $\phi$, we show that $\phi \in H^1_{\text{per}}(\Gamma_{\infty})$, and then use lemma~\ref{lem:lapgaminf}. For any $K \in \mathbb{R}$, we have from lemma~\ref{thm:denconvg} that $\phi_L$ weakly converges in $H^1(\Gamma_K)$ to $\phi \in H^1_{\text{loc}}(\Gamma_\infty)$. Since $\phi$ is a weak limit, it holds that
$$ \|\phi\|_{H^1(\Gamma_K)} \leq \liminf_{L \rightarrow \infty} \|\phi_L\|_{H^1(\Gamma_K)}.$$
However, for $L>K$, $\|\phi_L\|_{H^1(\Gamma_K)}^2 \leq \|\phi_L\|_{H^1(\Gamma_L)}^2 \leq C$, where $C$ is not dependent on $L$ or $K$. Hence $\|\phi\|_{H^1(\Gamma_K)} \leq \sqrt{C}$ for every $K$, thereby yielding $\phi \in H^1(\Gamma_\infty)$.
\end{proof}

Note that the Green's function used in this theorem is the same as in subsection~\ref{sec:thinfilm}; this will be assumed in the sequel unless stated otherwise.

\subsection{Convergence of energy difference on $\Gamma_\infty$}
Unlike the previous section dealing with the convergence of the electron density based on general known results, there is no general recipe to show the convergence of the relative energy of quasi-planar defects and/or the computation of the relative energy using a minimization principle. In this section, we demonstrate that both these are possible, and further establish an exponential rate of convergence with respect to cell size.
 
The quantity of interest for the present analysis is the difference in energy between the reference and perturbed configurations, with nuclear densities $m_1$ and $m_2$, respectively. Since these are both infinitely large, they cannot be compared directly. The strategy that we adopt here is to compute this difference in energy by recourse to a thermodynamic limit argument---we consider supercells of increasingly larger extent along the $z$ direction, and study the energy difference between $m_{1,L} \equiv m_1$ and $m_{2,L}$ in the limit $L \to \infty$. We clarify that this limit process only modifies the extent of the supercell along the $z$ direction, but retains it intact along the $x$ and $y$ directions. The periodicity along the $x$ and $y$ directions motivates the reference to the relative energy difference as an energy difference per unit area, similar to how surface energies are defined. In this section, we show that the difference in energies between these to configurations, which we denote as $\gamma^\nu_1$, is finite in the limit $L \to \infty$.

We begin by defining useful spaces of test functions for differences in the square root of the electron density between the reference and perturbed configurations: 
\begin{equation} \label{eq:test_functions_w}
\begin{split}
\mathcal{Q} &= \{w \in H^1_{per}(\Gamma_\infty) \;:\; |x_3|u_1w \in L^1(\Gamma_\infty), |x_3|w^2 \in L^1(\Gamma_\infty)\},\\
\mathcal{Q}_0 &= \{w \in H^1_{per}(\Gamma_\infty) \;:\; \exists\; K > 0 \text{ s.t. } \text{supp}(w\vert_{\Gamma_\infty})  \subset \bar{\Gamma}_{K} \},\\
\mathcal{Q}_1 &= \{w \in \mathcal{Q} \;:\; u_1 + w \geq 0\},\\
\mathcal{Q}_1^\nu &= \left\{w \in \mathcal{Q}_1 \;:\; \int_{\Gamma_\infty} (u_1 + w)^2 - u_1^2 = \int_{\Gamma_\infty}\nu\right\}. 
\end{split}
\end{equation}
We note for future reference that $\mathcal{Q}_0 \subset \mathcal{Q}$, $\mathcal{Q}^\nu_1 \subset \mathcal{Q}_1 \subset \mathcal{Q}$, $\mathcal{Q}_0 \cup \mathcal{Q}_1 \subset \mathcal{Q}$ and $w \in \mathcal{Q}$ implies that $|x_3|w \in L^1(\Gamma_\infty)$\footnote{Since $|x_3|u_1w \in L^1(\Gamma_\infty)$ and $u_1 \in L^\infty(\mathbb{R}^3)$}, $w \in L^1 \cap L^2(\Gamma_\infty)$\footnote{$\int_{\Gamma_\infty}|w| \leq \int_{\Gamma_3}|w| + \int_{\Gamma_\infty \setminus \Gamma_3}|x_3||w| \leq 3 \|w\|_{L^2(\Gamma_\infty)} + \||x_3|w\|_{L^1(\Gamma_\infty)}   $}, and $w\in L^p_{\text{unif}}(\mathbb{R}^3)$ for all $1 \leq p \leq 6$. With these definitions in place, we define a candidate function for the limiting energy difference between the reference and perturbed crystals, $\mathcal{E}^\nu_1:\mathcal{Q}_0 \cup \mathcal{Q}_1 \to \mathbb{R}$, as
\begin{equation} \label{eq:inffunc}
\begin{split}
\mathcal{E}^\nu_1(w) &= \int_{\Gamma_\infty} \lVert\nabla w\rVert^2 + \int_{\Gamma_\infty} (|u_1 + w|^{10/3}-u_{1}^{10/3} - \frac{10}{3}u_1^{7/3}w)\\ 
 & + \int_{\Gamma_\infty} \phi_{1} (w^2 - \nu) + \frac{1}{2}D_G((u_1 + w)^2 - u_1^2 -\nu, (u_1 + w)^2 - u_1^2 -\nu).
\end{split}
\end{equation}
In equation \eqref{eq:inffunc}, the function $D_G:\mathcal{Q}\times\mathcal{Q}\to\mathbb{R}$ is defined as follows: for $f, g \in \mathcal{Q}$,
\begin{equation} \label{eq:DG_Gamma_infty}
D_G(f,g) = \int_{\Gamma_\infty} \int_{\Gamma_\infty} G(x-y)f(x) g(y) \,dx dy.
\end{equation}
Section~\ref{subsec:thinfilm} ensures that $D_G$ is a well defined function and lemma~\ref{lem:KEhelp} ensures that $\mathcal{E}^\nu_1$ in equation \eqref{eq:inffunc} is a also well defined function. We define the difference in energy between the perturbed system and the reference perfect crystal as in equation~\eqref{eq:gamma_L_defn}, reproduced below:
\begin{equation*}
\gamma^L_1(\nu) = I^L_{m_1 + \nu} - I^L_{m_1}.
\end{equation*}
Here, $I^L_m$ is defined as in equation~ (\ref{eq:TFW_min_supercell}). The following theorem proves that this difference in energy remains finite in the thermodynamic limit $L \to \infty$; we notate the limiting energy difference as $\gamma^\nu_1$.

\begin{theorem} \label{thm:convggam}
$\gamma^L_1(\nu)$ converges to $\gamma^\nu_1$ in the limit $L \to \infty$, where
\begin{equation} \label{eq:gamma_m1_nu}
\begin{split}
\gamma^\nu_1 &= \int_{\Gamma_\infty} \lVert\nabla v\rVert^2 + \int_{\Gamma_\infty} (u_1 + v)^{10/3}- u_{1}^{10/3} - \frac{10}{3}u_1^{7/3}v\\  
 &+ \frac{1}{2}\int_{\Gamma_\infty} \phi ((u_1 + v)^2 - u_1^2 -\nu) + \int_{\Gamma_\infty} \phi_{1} (v^2  - \nu).
\end{split}
\end{equation}
\end{theorem}
\begin{proof}
We have, using equation~\eqref{eq:TFW_energy_supercell} and Theorem~\ref{thm:TFW_supercell_exst_uniq}, that
\begin{displaymath}
\begin{split}
    I^L_{m_2} &= \int_{\Gamma_L} \lVert\nabla u_{2,L}\rVert^2 + \int_{\Gamma_L} |u_{2,L}|^{\frac{10}{3}}  + \frac{1}{2} \int_{\Gamma_L} \phi_{2,L}(u_{2,L}^2 - m_2), \\
    I^L_{m_1} &= \int_{\Gamma_L} \lVert\nabla u_{1}\rVert^2 + \int_{\Gamma_L} |u_{1}|^{\frac{10}{3}}  + \frac{1}{2} \int_{\Gamma_L} \phi_{1}(u_{1}^2-m_1).\\
\end{split}
\end{displaymath}
Hence, using equation \eqref{eq:TFW_smeared_EL} corresponding to $n = m_1$, we get
\begin{displaymath}
\begin{split}
\gamma^L_1(\nu) &= \int_{\Gamma_L} \lVert\nabla v_L\rVert^2 + \int_{\Gamma_L} (u_1+v_L)^{10/3}-u_{1}^{10/3} - \frac{10}{3}u_1^{7/3}v_L - 2 \int_{\Gamma_L}\phi_1 u_1v_L \\
  &+ \frac{1}{2}\int_{\Gamma_L} \phi_L (u_1^2 + v_L^2 +2u_1v_L - m_1 -\nu)  + \frac{1}{2}\int_{\Gamma_L} \phi_{1} (v_L^2 + 2u_1v_L - \nu). 
\end{split}
\end{displaymath}
Since $-\nabla^2 \phi_L = 4\pi (v_L^2 +2u_1v_L -\nu)$ and $-\nabla^2 \phi_1 = 4\pi (u_1^2 - m_1)$, by definition, we obtain, after integrating by parts twice, that
\begin{displaymath}
\int_{\Gamma_L} \phi_L (u_1^2 - m_1 ) = \int_{\Gamma_L} \phi_1 (2u_1v_L + v_L^2-\nu).    
\end{displaymath}
Therefore, 
\begin{displaymath}
\begin{split}
\gamma^L_1(\nu) &= \int_{\Gamma_L} \lVert\nabla v_L\rVert^2 + \int_{\Gamma_L} (u_1+v_L)^{10/3}-u_{1}^{10/3} - \frac{10}{3}u_1^{7/3}v_L\\ 
  &+ \frac{1}{2}\int_{\Gamma_L} \phi_L ( v_L^2 +2u_1v_L -\nu)    + \int_{\Gamma_L} \phi_{1} (v_L^2  - \nu). 
\end{split}
\end{displaymath}
Corollary~\ref{cor:supercell_field_convergence} informs us that, in the limit $L \to \infty$,
\begin{align*}
    \chi_{\Gamma_L}v_L &\rightarrow \chi_{\Gamma_\infty} v, \\
    \chi_{\Gamma_L}\nabla v_L &\rightarrow \chi_{\Gamma_\infty} \nabla v, \\
    \chi_{\Gamma_L}\phi_L &\rightarrow \chi_{\Gamma_\infty} \phi,
\end{align*}
in $L^p(\Gamma_\infty)$ for all $1 \leq p < \infty$. From Lemma~\ref{thm:density_decay} we also have that $v_L, \|\nabla v_L\|, $ and $\phi_L $ are uniformly bounded irrespective of $L$. Moreover, $v, \|\nabla v\|,$ and $\phi$ belong to $L^\infty(\mathbb{R}^3)$.

We therefore have, using lemma~\ref{lem:DCThelp}, that
\begin{displaymath}
\int_{\Gamma_L} (u_1+v_L)^{10/3}-u_{1}^{10/3} - \frac{10}{3}u_1^{7/3}v_L  \rightarrow \int_{\Gamma_\infty} (u_1+v)^{10/3}-u_{1}^{10/3} - \frac{10}{3}u_1^{7/3}v, 
\end{displaymath}
in the limit $L \to \infty$.

Finally, a direct application of Hölder's inequality yields
\begin{displaymath}
\frac{1}{2}\int_{\Gamma_L} \phi_L (v_L^2 + 2u_1v_L - \nu) + \int_{\Gamma_L} \phi_1 (v_L^2 - \nu) \to \frac{1}{2}\int_{\Gamma_\infty} \phi \big((u_1 + v)^2 - u_1^2 - \nu\big) + \int_{\Gamma_\infty} \phi_1 (v^2 - \nu),
\end{displaymath}
as \( L \to \infty \). The boundedness condition \( \phi_1 \in L^\infty(\mathbb{R}^3) \) is essential for the convergence of the second term, which justifies our initial assumption that \( m_1 \) is smooth. We thus conclude that \( \gamma^L_1(\nu) \to \gamma^\nu_1 \), as desired.
 For completeness, we note that $\gamma^\nu_1$ is $\mathcal{E}_1^\nu(v)$, thanks to corollary~\ref{cor:establish_green_function}.
\end{proof}

\subsection{Minimization principle for the energy difference on $\Gamma_\infty$}
\label{sec: minproblem}

We show in this section that the converged energy difference $\gamma^\nu_1$ that we highlighted in the previous section satisfies a minimization principle. For this purpose, we resort to the thin films approach summarized in Section~\ref{sec:thinfilm}; we remark that we chose this approach over the approach involving supercells (Section~4.7 in \cite{CE11}) since the latter was not amenable to the derivation of a minimization principle.

Let $H \in \mathbb{N}$ such that $\Omega_H = \mathbb{R}^2 \times (-\frac{H}{2}, \frac{H}{2})$, contains $\text{supp}(\nu)$. We define $\Tilde{m}_{1,H} = \chi_{\Omega_H}m_1$ and $\Tilde{m}_{2,H} = \chi_{\Omega_H}m_2$. A visual representation of $\Tilde{m}_{2,H}$ is shown in Figure~\ref{fig:thinfilm_construction}. We further assume in this section that $m_1$ is smooth and symmetric in $x_1$ and $m_2$ is also symmetric along $x_1$; we do not, however, require $m_2$ to be  smooth. The arguments we present below also hold if $m_1$ is symmetric along $x_2$, instead of $x_1$. In this context, we present in Section~\ref{subsec:thinfilmmod} a slight modification of the existence proof for the thermodynamic limit for thin films, originally presented in \cite{BTF00}, that relaxes the symmetry requirement on $m_1$ as required above. 

\begin{figure}[h]
\begin{subfigure}{0.45\textwidth}
\centering
\includegraphics[width=0.9\textwidth]{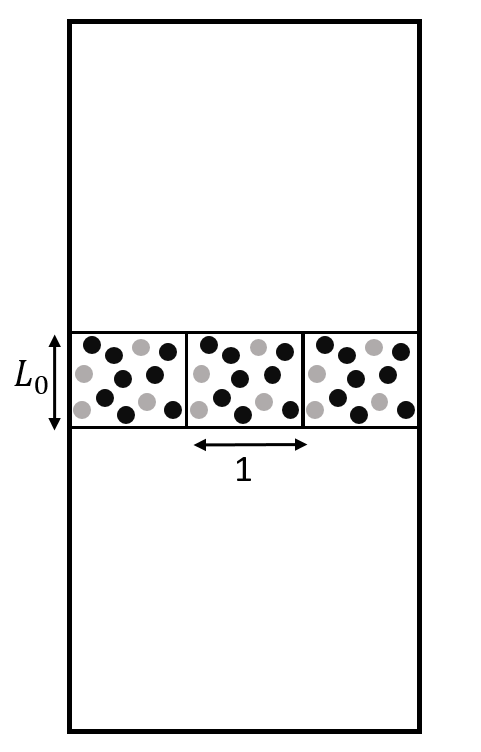}
\label{fig:m_tf}
\caption{$m$}
\end{subfigure}
~
\begin{subfigure}{0.49\textwidth}
\centering
\includegraphics[width=0.85\textwidth]{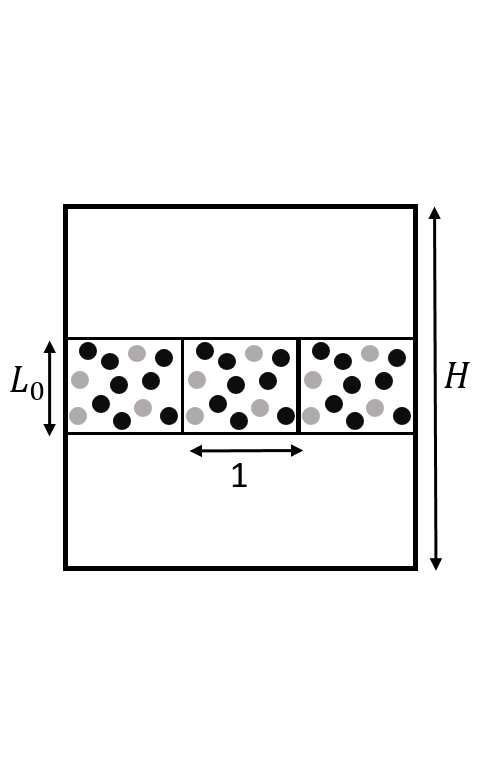}
\label{fig:mH}
\caption{$\Tilde{m}_H$}
\end{subfigure}
\caption{Procedure to generate a thin film with nuclear density $\Tilde{m}_H$ from a given nuclear density $m \in \mathcal{M}_{xy}$. The original system is sliced into a box with dimension $H$ along the out-of-plane direction of the planar defect.}
\label{fig:thinfilm_construction}
\end{figure}

We know from the thin films theory that there exist unique electron densities $\Tilde{u}_{1,H}$ and $\Tilde{u}_{2,H}$ for both these systems; the corresponding potentials are denoted as $\psi_{1,H}$ and $\psi_{2,H}$. 
It also holds that $(\Tilde{u}_{1,H}, \psi_{1,H})$ and $(\Tilde{u}_{2,H}, \psi_{2,H})$ are bounded uniformly in $H^4_{\text{unif}}(\mathbb{R}^3) \times H^2_{\text{unif}}(\mathbb{R}^3)$. Taking the local limit as $H \to \infty$ for the thin films case (equation~\eqref{eq:ELthinfilm2}), we recover the corresponding Schrodinger-Poisson system on $\Gamma_\infty$ (equation~\eqref{eq:TFW_smeared_EL}). Since the latter admits a unique solution, $\Tilde{u}_{1,H}, \psi_{1,H}, \Tilde{u}_{1,H}$ and $\psi_{2,H}$ converge to $u_1, \phi_1, u_2$ and $\phi_2$, respectively, pointwise almost everywhere, along a subsequence.

In analogy with the analysis in the supercell approach, we define $\Tilde{v}_{H} = \Tilde{u}_{2,H}-\Tilde{u}_{1,H}$ and $\psi_H=\psi_{2,H}-\psi_{1,H}$. We know that they converge to $v$ and $\phi$ pointwise almost everywhere, along a subsequence. Moreover, we also have that $\psi_{1,H}$ is uniformly bounded almost everywhere since $m_1$ is smooth. With these definitions in place, we state the central result of this section, namely that the thermodynamic limit of the difference in energy between the reference and perturbed systems satisfies a minimum principle.

\begin{theorem}
\label{thm:reflexivemin}
$v = u_2 - u_1$ is the minimizer of the following variational principle:
\begin{equation} \label{eq:gamma_min_ppl}
\inf_{w \in \mathcal{Q}_1^{\nu}} \; \mathcal{E}^\nu_1(w), 
\end{equation}
where 
\begin{displaymath}
\begin{split}
\mathcal{E}^\nu_1(w) &= \int_{\Gamma_\infty} \lVert\nabla w\rVert^2 + \int_{\Gamma_\infty} (|u_1 + w|^{10/3}-u_{1}^{10/3} - \frac{10}{3}u_1^{7/3}w)\\ 
 & + \int_{\Gamma_\infty} \phi_{1} (w^2 - \nu) + \frac{1}{2}D_G((u_1 + w)^2 - u_1^2 -\nu, (u_1 + w)^2 - u_1^2 -\nu),\\
\end{split}
\end{displaymath}  
The minimum value is equal to the thermodynamic limit of the energy difference between the reference and perturbed configurations:
\begin{equation} \label{eq:gamma_min_energy}
\gamma^\nu_1 = \inf_{w \in \mathcal{Q}_1^{\nu}} \; \mathcal{E}^\nu_1(w).
\end{equation}
\end{theorem}
\begin{remark}
Note that $v \in \mathcal{Q}_1^{\nu}$. As a consequence of this theorem, we define the surface energy density as $\gamma^\nu_1 = \mathcal{E}^\nu_1(v) = \mathcal{E}^\nu_1(u_2-u_1)$.
\end{remark}

\begin{proof}
We present first an overview of the proof first, for the sake of readability, and subsequently provide all the relevant details in multiple step. Since $v \in \mathcal{Q}^\nu_1$, the strategy we adopt to prove the minimization principle is by showing that $\mathcal{E}^\nu_1(v) \le \mathcal{E}^\nu_1(w)$ for every $w \in \mathcal{Q}^\nu_1$. We employ the thin films theory, as mentioned previously, to establish this via the following steps:
\begin{enumerate}[(i)]
\item \emph{Step-1:} Given $w \in \mathcal{Q}^\nu_1$, construct a one-parameter family of functions $w_\epsilon$ in $\mathcal{Q}_1$, with $\epsilon \in \mathbb{R}$, such that, along a subsequence,
\begin{displaymath}
\lim_{\epsilon \rightarrow 0^+}\mathcal{E}^\nu_1(w_\epsilon) = \mathcal{E}^\nu_1(w).
\end{displaymath}
\item \emph{Step-2:} Starting from $w_\epsilon$, construct a two-parameter family of functions $w_{H,\epsilon}$ in $\mathcal{Q}_0$ such that the family of functions $\Tilde{u}_{1,H} + w_{H,\epsilon}$ are also candidate solutions for $\Tilde{m}_{2,H}$, and prove the following inequality:
\begin{displaymath}
\limsup_{H \rightarrow \infty} \mathcal{E}^\nu_1(w_{H, \epsilon}) \leq  \mathcal{E}^\nu_1(w_\epsilon).
\end{displaymath}
\item \emph{Step-3:} Defining $\mathcal{E}^\nu_H(w) = E_H^{\Tilde{m}_{2,H}}(\Tilde{u}_{1,H} + w) - E_H^{\Tilde{m}_{1,H}}(\Tilde{u}_{1,H})$, where $E^m_H$ is the thin films energy \eqref{eq:thinfilm_green}, show that
\begin{displaymath}
\liminf_{H \rightarrow \infty} \mathcal{E}^\nu_H(w_{H,\epsilon}) \leq  \limsup_{H \rightarrow \infty} \mathcal{E}^\nu_1(w_{H,\epsilon}) < \infty.
\end{displaymath} 
\item \emph{Step-4:} Defining $\Tilde{v}_{H} = \Tilde{u}_{2,H} - \Tilde{u}_{1,H}$ in terms of the solutions $\Tilde{u}_{1,H}$ and $\Tilde{u}_{2,H}$ of the thin films problem with nuclear densities $\Tilde{m}_{1,H}$ and $\Tilde{m}_{2,H}$, respectively, establish the following inequality:
\begin{displaymath}
\mathcal{E}^\nu_1(v) \leq \liminf_{H \rightarrow \infty} \mathcal{E}^\nu_H(\Tilde{v}_{H}).
\end{displaymath}
\end{enumerate}
The proof of the theorem follows by combining the conclusion of the foregoing steps. Indeed, since 
\begin{displaymath}
\mathcal{E}^\nu_H(\Tilde{v}_{H}) \le \mathcal{E}^\nu_H(w_{H,\epsilon}).
\end{displaymath}
We therefore see that
\begin{displaymath}
\begin{split}
\gamma^\nu_1 = \mathcal{E}^\nu_1(v) &\le \liminf_{H \rightarrow \infty} \mathcal{E}^\nu_H(\Tilde{v}_{H})\\
 &\le \liminf_{H \rightarrow \infty}  \mathcal{E}^\nu_H(w_{H,\epsilon})\\
 &\le \limsup_{H \rightarrow \infty} \mathcal{E}^\nu_1(w_{H,\epsilon})\\
 &\le  \mathcal{E}^\nu_1(w_{\epsilon}) .
\end{split}
\end{displaymath}
Taking the limit $\epsilon \to 0^{+}$ establishes the theorem. 

We present next detailed proofs of the individual steps. 

\emph{Step-1:} Define the function $f_\epsilon:\mathbb{R} \to \mathbb{R}$, for any $\epsilon > 0$, as
$$f_\epsilon(t) = \begin{cases}
1, & |t| \leq 1/\epsilon,\\ 
(-\epsilon t + 2), & 1/\epsilon < t < 2/\epsilon, \\
(\epsilon t + 2), & -2/\epsilon < t < -1/\epsilon, \\
0, & \text{otherwise.}
\end{cases} $$
Given $w \in \mathcal{Q}_1^\nu$, define $w_\epsilon(x) = f_\epsilon(x_3)w(x)$. It is easy to see that $|w_{\epsilon}| \leq |w|$,  $w_{\epsilon} \in H^1_{per}(\Gamma_\infty)$, $\text{supp}(\chi_{\Gamma_\infty}w_{\epsilon}) \subset  \Gamma_{\frac{4}{\epsilon}}$ and hence $w_\epsilon \in \mathcal{Q}_0$. Since $w > -u_1$, we also have $w_{\epsilon} \geq -u_1$, and hence, $w_\epsilon \in \mathcal{Q}_1$. Since 
\begin{displaymath}
    \begin{split}
               \|w - w_{\epsilon}\|_{H^1(\Gamma_\infty)}^2 &= 
        \| w - w_{\epsilon}\|_{H^1(\Gamma_{4/\epsilon} \setminus \Gamma_{2/\epsilon})}^2 + \| w\|_{H^1(\Gamma_\infty \setminus \Gamma_{4/\epsilon})}^2 \\
        &= \|w(1 - f_\epsilon(x_3))\|^2_{L^2(\Gamma_{4/\epsilon} \setminus \Gamma_{2/\epsilon})} + \| |\nabla (w(1 - f_\epsilon(x_3)))| \|^2_{L^2(\Gamma_{4/\epsilon} \setminus \Gamma_{2/\epsilon})} \\& + \| w\|_{H^1(\Gamma_\infty \setminus \Gamma_{4/\epsilon})}^2\\
        &\leq \| w \|_{L^2(\Gamma_{4/\epsilon} \setminus \Gamma_{2/\epsilon})}^2 + 2\| |\nabla w| \|_{L^2(\Gamma_{4/\epsilon} \setminus \Gamma_{2/\epsilon})}^2 + 2 \| \epsilon w \|_{L^2(\Gamma_{4/\epsilon} \setminus \Gamma_{2/\epsilon})}^2 \\ &+ \| w\|_{H^1(\Gamma_\infty \setminus \Gamma_{4/\epsilon})}^2 \\
                         &\leq C\| w \|_{H^1(\Gamma_{\infty} \setminus \Gamma_{2/\epsilon})}^2 + 
        2\epsilon^2 \| w \|_{L^2(\Gamma_\infty)},
    \end{split}
\end{displaymath}
for some constant $C > 0$, we see that $w_{\epsilon}$ converges to $w$ in $H^1(\Gamma_\infty)$ as $\epsilon$ goes to $0$. Similarly we have convergence in $L^p(\Gamma_\infty)$ when $w$ is in $L^p(\Gamma_\infty)$ for all $1 \le p < \infty$. Since $|w_\epsilon| \leq |w|$ and $w_\epsilon$ converges pointwise almost everywhere to $w$ along a subsequence, we can use Dominated Convergence Theorem (DCT) along with Lemma~\ref{lem:KEhelp} to show that 
$$ \int_{\Gamma_\infty} (|u_1 + w_\epsilon|^{10/3}-u_{1}^{10/3} - \frac{10}{3}u_1^{7/3}w_\epsilon) \rightarrow \int_{\Gamma_\infty} (|u_1 + w|^{10/3}-u_{1}^{10/3} - \frac{10}{3}u_1^{7/3}w).$$
The bound $|w_\epsilon| \leq |w|$ can further be used with DCT to yield
\begin{equation}
   \label{eq:minproblem_local1}
   \lim_{\epsilon \rightarrow 0^+}\mathcal{E}^\nu_1(w_\epsilon) = \mathcal{E}^\nu_1(w),
\end{equation} 
along a subsequence. We will henceforth work with this subsequence, and denote it as $w_\epsilon$. 

\emph{Step-2:} The one-parameter family of functions $w_\epsilon$ constructed in Step-1 may not satisfy the charge neutrality constraint. To rectify this, we construct a two-parameter family of functions $w_{H,\epsilon}$, starting from $w_\epsilon$, as follows. Define 
$$q_\epsilon = \int_{\Gamma_\infty}((u_1 + w_\epsilon)^2 - u_1^2 -\nu).$$ 
Since $w \in \mathcal{Q}_1^\nu$, we see that
\begin{displaymath}
\begin{split}
\frac{|q_\epsilon|}{\epsilon} &=  \frac{1}{\epsilon}\Bigg|\int_{\Gamma_\infty} w_\epsilon^2 + 2u_1w_\epsilon - w^2 -2u_1w \Bigg| \\ 
&\leq \frac{1}{\epsilon} \Big| \int_{\Gamma_\infty} w_\epsilon^2 - w^2 \Big| + \frac{1}{\epsilon} \Big|\int_{\Gamma_\infty} 2u_1(w_\epsilon-w) \Big| \\
   &= \frac{1}{\epsilon}\Big| \int_{\Gamma_{\frac{4}{\epsilon}} \setminus \Gamma_{\frac{2}{\epsilon}}} (w_\epsilon^2 - w^2) - \int_{\Gamma_{\infty} \setminus \Gamma_{\frac{4}{\epsilon}}}w^2\Big| + 
   \frac{1}{\epsilon}\Big| \int_{\Gamma_{\frac{4}{\epsilon}} \setminus \Gamma_{\frac{2}{\epsilon}}}  2u_1(w_\epsilon - w) - \int_{\Gamma_{\infty} \setminus \Gamma_{\frac{4}{\epsilon}}}2u_1w \Big| \\
    &\le \frac{1}{\epsilon} \int_{\Gamma_{\frac{4}{\epsilon}} \setminus \Gamma_{\frac{2}{\epsilon}}} |w_\epsilon^2 - w^2|  + \frac{1}{\epsilon} \int_{\Gamma_{\infty} \setminus \Gamma_{\frac{4}{\epsilon}}}w^2 + 
   \frac{1}{\epsilon} \int_{\Gamma_{\frac{4}{\epsilon}} \setminus \Gamma_{\frac{2}{\epsilon}}} 2u_1|w_\epsilon - w|  + \frac{1}{\epsilon}\int_{\Gamma_{\infty} \setminus \Gamma_{\frac{4}{\epsilon}}}2u_1|w| \\
   &\leq  \frac{1}{\epsilon} \int_{\Gamma_{\infty} \setminus \Gamma_{\frac{2}{\epsilon}}}w^2 + 
     \frac{1}{\epsilon}\int_{\Gamma_{\infty} \setminus \Gamma_{\frac{2}{\epsilon}}}2u_1|w| \\
   &\leq   \int_{\Gamma_{\infty} \setminus \Gamma_{\frac{2}{\epsilon}}}|x_3|w^2 + 
     \int_{\Gamma_{\infty} \setminus \Gamma_{\frac{2}{\epsilon}}}2u_1|x_3||w|. 
\end{split}
\end{displaymath}
We thus get, as a consequence of DCT, that 
\begin{equation} \label{eq:qepsilon_limit}
   \lim_{\epsilon \rightarrow 0^+} \frac{q_\epsilon}{\epsilon} =0. 
\end{equation}
For $H, \epsilon \in \mathbb{R}^+$, we follow a procedure analogous to that used in Lemma~4.3 in ~\cite{CE11} and construct
$$w_{H,\epsilon}(x)  = w_\epsilon(x) + t_{H,\epsilon} g_\epsilon(x_3) \Tilde{u}_{1,H}(x),$$
where $g_\epsilon:\mathbb{R}\to\mathbb{R}$ is defined as
$$g_\epsilon(t) =  f_{\epsilon}(t - (4/\epsilon + 1)).$$ In the equation displayed above, $t_{H,\epsilon} \in \mathbb{R}$ is a constant that is chosen to obtain the charge neutrality of $w_{H,\epsilon}$. By construction, $w_\epsilon(x)$ and $t_{H,\epsilon} g_\epsilon(x_3) \Tilde{u}_{1,H}(x)$ are supported on disjoint sets, the distance between which is at least $1$. Therefore, 
$$\int_{\Gamma_\infty}w_{H, \epsilon}^2 + 2\Tilde{u}_{1,H}w_{H, \epsilon} - \nu = t_{H,\epsilon}^2 \int_{\Gamma_\infty}g_\epsilon^2\Tilde{u}_{1,H}^2 + 2t_{H,\epsilon}\int_{\Gamma_\infty}g_\epsilon \Tilde{u}_{1,H}^2 + 2 \int_{\Gamma_\infty}(\Tilde{u}_{1,H}-u_1)w_\epsilon + q_\epsilon.$$ 
For charge neutrality of $w_{H,\epsilon}$, we require the expression shown above to vanish. The quadratic equation for $t_{H,\epsilon}$ has real roots if 
$$ \left(\int_{\Gamma_\infty}g_\epsilon \Tilde{u}_{1,H}^2\right)^2 \geq \left(\int_{\Gamma_\infty}g_\epsilon^2\Tilde{u}_{1,H}^2\right)\left(2 \int_{\Gamma_\infty}(\Tilde{u}_{1,H}-u_1)w_\epsilon + q_\epsilon\right).$$
Consider the term
$$\frac{(\int_{\Gamma_\infty}g_\epsilon \Tilde{u}_{1,H}^2)^2}{(\int_{\Gamma_\infty}g_\epsilon^2\Tilde{u}_{1,H}^2)} - \left(2 \int_{\Gamma_\infty}(\Tilde{u}_{1,H}-u_1)w_\epsilon + q_\epsilon\right).$$ 
It is well defined since according to Theorem~\ref{thm:thinfilm}, for any $R> 0$ there exists a positive constant $\nu_R$ such that $\inf_{|x_3|< R} \Tilde{u}_{1,H} \geq \nu_R$.
Noting that $\Tilde{u}_{1,H}$ converges to $u_1$ in the limit $H \to \infty$ (see Appendix~\ref{subsec:thinfilm}), and $q_\epsilon \to 0$ in the limit $\epsilon \to 0$, we see that
\begin{displaymath}
\lim_{\epsilon \to 0} \lim_{H \to \infty} \frac{(\int_{\Gamma_\infty}g_\epsilon \Tilde{u}_{1,H}^2)^2}{(\int_{\Gamma_\infty}g_\epsilon^2\Tilde{u}_{1,H}^2)} - \left(2 \int_{\Gamma_\infty}(\Tilde{u}_{1,H}-u_1)w_\epsilon + q_\epsilon\right) = \lim_{\epsilon \to 0} \lim_{H \to \infty} \frac{(\int_{\Gamma_\infty}g_\epsilon \Tilde{u}_{1,H}^2)^2}{(\int_{\Gamma_\infty}g_\epsilon^2\Tilde{u}_{1,H}^2)}.
\end{displaymath}
 Hence we see that the limit displayed above is positive for large enough $H$ and small enough $\epsilon$. We therefore get two real values of $t_{H,\epsilon}$ such that, for large $H$ and small $\epsilon$,
$$\int_{\Gamma_\infty}w_{H,\epsilon}^2 + 2\Tilde{u}_{1,H}w_{H,\epsilon} = \int_{\Gamma_\infty} \nu.$$ 
Fix the larger of these two values as $t_{H,\epsilon}$. For small enough $\epsilon$, we then have,
\begin{align*}
    \lim_{H \rightarrow \infty}t_{H, \epsilon} &= 
    \frac{- \int_{\Gamma_\infty}g_\epsilon u_{1}^2 + \sqrt{  \left(\int_{\Gamma_\infty}g_\epsilon u_{1}^2\right)^2 - \left(\int_{\Gamma_\infty}g_\epsilon^2u_{1}^2\right) q_\epsilon}} {\int_{\Gamma_\infty}g_\epsilon^2u_{1}^2} \\
    & = - z_\epsilon + z_\epsilon\sqrt{1 - \frac{q_\epsilon} {z_\epsilon^2\int_{\Gamma_\infty} g_\epsilon^2u_{1}^2}},
\end{align*}
where $z_\epsilon = \frac{ \int_{\Gamma_\infty}g_\epsilon u_{1}^2}{\int_{\Gamma_\infty}g_\epsilon^2u_{1}^2}$ and is of the order of $\frac{1 +\epsilon}{1+\epsilon^2}$. Since $q_\epsilon$ is at least of the order of $\epsilon$ from equation~\eqref{eq:qepsilon_limit} we use the binomial theorem for small enough $\epsilon$ to get
$$\lim_{H \rightarrow \infty}|t_{H, \epsilon}| = C\frac{|q_\epsilon|}{\int_{\Gamma_\infty} g_\epsilon^2} + O(\epsilon^4) \leq C|q_\epsilon| \epsilon + O(\epsilon^4).$$ Here, $C$ is a constant that depends only on the uniform bounds of $m_1$.

The two-parameter family of functions $w_{H,\epsilon}$ constructed above thus satisfies the charge neutrality condition for large enough $H$ and small enough $\epsilon$. Moreover, $ w_{H, \epsilon} \in H^1_{per}(\Gamma_\infty)$ and is compactly supported in $x_3$ direction, hence belongs to the space $X^{\Tilde{m}_{2,H}}$, as defined in equation~\eqref{eq:thinfilm_green}.

It is easy to see that $w_{H, \epsilon} \in \mathcal{Q}_0$. Therefore, 
\begin{equation}
    \label{eq:constructwhe_loc1}
    \mathcal{E}^\nu_1(w_{H, \epsilon}) = \mathcal{E}^\nu_1(w_\epsilon) + \mathcal{E}_{m_1}^0(t_{H,\epsilon} g_\epsilon \Tilde{u}_{1,H}) + D_G\Big(w_\epsilon^2 + 2u_1w_\epsilon - \nu, t_{H,\epsilon} g_\epsilon \Tilde{u}_{1,H} \Big).
\end{equation}

Choosing $\epsilon$ and $H$ such that $t_{H,\epsilon} \leq 1$, we get 
\begin{displaymath}
 |\mathcal{E}_{1}^0(t_{H,\epsilon} g_\epsilon \Tilde{u}_{1,H})| \leq \frac{C_1 t_{H,\epsilon}}{\epsilon} + C_1 D_G(t_{H,\epsilon}^2 g_\epsilon^2  + t_{H,\epsilon} g_\epsilon, t_{H,\epsilon}^2 g_\epsilon^2  + t_{H,\epsilon} g_\epsilon),     
\end{displaymath}
where $C_1 > 0$ is  only dependent on uniform bounds of $m_1$. Now since $|G(x) + 2\pi |x_3| - \frac{1}{|x|}| \in L^\infty(\mathbb{R}^3)$ we get, 
$$|D_G(t_{H,\epsilon}^2 g_\epsilon^2  + t_{H,\epsilon} g_\epsilon, t_{H,\epsilon}^2 g_\epsilon^2  + t_{H,\epsilon} g_\epsilon)| \leq C_2\left( \frac{t^2_{H,\epsilon}}{\epsilon^2} + \frac{4}{\epsilon}\frac{t^2_{H,\epsilon}}{\epsilon^2} +  \frac{t^2_{H,\epsilon}}{\epsilon^2} \right).$$
Here $C_2 >0$, doesn't depend on $\epsilon$ or $H$. Note that we have used the fact that $|x_3-y_3| \leq \frac{4}{\epsilon}$ and $\int_{\Gamma_K}\int_{\Gamma_K} \frac{1}{|x-y|} \leq C K^2$, for some constant $C$. We thus obtain
\begin{equation}
    \label{eq:constructwhe_loc2}
    \limsup_{\epsilon \rightarrow 0^+} \limsup_{H \rightarrow \infty} \, \mathcal{E}_{1}^0(t_{H,\epsilon} g_\epsilon \Tilde{u}_{1,H}) =0.
\end{equation}
Further, use the fact that the supports of $w_\epsilon$ and $g_\epsilon$ are separated by a distance which is at least larger than $1$, we get  
$$\lvert D_G(w_\epsilon^2 + 2u_1w_\epsilon - \nu, t_{H,\epsilon} g_\epsilon \Tilde{u}_{1,H}) \rvert \leq C_3 \left( \frac{t_{H,\epsilon}}{\epsilon} + \left(\frac{8}{\epsilon} + 1\right)\frac{t_{H,\epsilon}}{\epsilon} \right)\int_{\Gamma_\infty}(w_\epsilon^2 + 2u_1w_\epsilon - \nu).$$
In deriving the bounds above, we have used that $|x_3-y_3| \leq \frac{8}{\epsilon} + 1$. Note that all terms will be at least of the order of $\epsilon$ except $t_{H,\epsilon}{\epsilon^2}$. Using $\lim_{\epsilon \rightarrow 0^+} \frac{q_\epsilon}{\epsilon} = 0$ we get
\begin{equation}
   \label{eq:constructwhe_loc3}
    \limsup_{\epsilon \rightarrow 0^+} \limsup_{H \rightarrow \infty} D_G\Big(w_\epsilon^2 + 2u_1w_\epsilon - \nu, t_{H,\epsilon} g_\epsilon \Tilde{u}_{1,H} \Big) \leq 0.
\end{equation}
Hence using equations~\eqref{eq:constructwhe_loc1}, \eqref{eq:constructwhe_loc2}, and \eqref{eq:constructwhe_loc3} we get, 
$$\limsup_{\epsilon \rightarrow 0^+}\limsup_{H \rightarrow \infty} \mathcal{E}^\nu_1(w_{H, \epsilon}) \leq \limsup_{\epsilon \rightarrow 0^+} \mathcal{E}^\nu_1(w_\epsilon).$$
Hence, we get from equation~\eqref{eq:minproblem_local1} that 
\begin{equation}
   \label{eq:minproblem_local2}
   \limsup_{\epsilon \rightarrow 0^+}\limsup_{H \rightarrow \infty}\mathcal{E}^\nu_1(w_{H,\epsilon}) \leq \mathcal{E}^\nu_1(w).
\end{equation}

\emph{Step-3:} Define $$\mathcal{E}^\nu_H(w) = E_H^{\Tilde{m}_{2,H}}(\Tilde{u}_{1,H} + w) - E_H^{\Tilde{m}_{1,H}}(\Tilde{u}_{1,H}),$$ where $E_H^{m}$ is defined as in Section~\ref{subsec:thinfilm}, for all $w$ such that $\Tilde{u}_{1,H} + w \in X^{\Tilde{m}_{2,H}}$. Using equation~\eqref{eq:thin_film_EL}, we can rewrite the foregoing equation as 
\begin{displaymath}
\begin{split}
\mathcal{E}^\nu_H(u) &= \int_{\Gamma_\infty} \lVert\nabla u\rVert^2 + \int_{\Gamma_\infty} (|u+\Tilde{u}_{1,H}|^{10/3}-\Tilde{u}_{1,H}^{10/3} - \frac{10}{3}\Tilde{u}_{1,H}^{7/3}u)  + \int_{\Gamma_\infty} \psi_{1,H} (u^2 - \nu_H) \\ 
& + \frac{1}{2}D_G(u^2 +2\Tilde{u}_{1,H}u -\nu_H, u^2 +2\Tilde{u}_{1,H}u -\nu_H),
\end{split}
\end{displaymath}
where $\nu_H = \Tilde{m}_{2,H}-\Tilde{m}_{1,H}$. Note that the terms involving the surface integrals when using the divergence theorem vanish since $u \in H^1_{\text{per}}(\Gamma_\infty)$.  

Since $\Tilde{u}_{1,H} + w_{H,\epsilon} \in X^{\Tilde{m}_{2,H}}$ and $\nu_H$ is compactly supported independently of $H$, it is trivial to see that $\lim_{H \rightarrow \infty} |\mathcal{E}^\nu_H(w_{H,\epsilon}) - \mathcal{E}^\nu_1(w_{H,\epsilon})| = 0 $. Therefore, 
\begin{equation}
    \label{eq:minproblem_local3}
    \liminf_{H \rightarrow \infty} \mathcal{E}^\nu_H(w_{H,\epsilon}) \leq  \limsup_{H \rightarrow \infty} \mathcal{E}^\nu_1(w_{H,\epsilon}) < \infty.
\end{equation}

\emph{Step-4:} We have $\mathcal{E}^\nu_H(\Tilde{v}_{H}) = E_H^{\Tilde{m}_{2,H}}(\Tilde{v}_{H} + \Tilde{u}_{1,H}) - E_H^{\Tilde{m}_{1,H}}(\Tilde{u}_{1,H})$, and since $E_H^{\Tilde{m}_{2,H}}(\Tilde{v}_{H} + \Tilde{u}_{1,H}) \leq E_H^{\Tilde{m}_{2,H}}(\Tilde{u}_{1,H} + w_{H,\epsilon})$ we get
\begin{equation}
   \label{eq:minproblem_local4}
     \mathcal{E}^\nu_H(\Tilde{v}_{H}) \leq \mathcal{E}^\nu_H(w_{H, \epsilon}). 
\end{equation}
Let us write for convenience $\mathcal{E}^\nu_H(\Tilde{v}_{H})$, using neutrality, as
\begin{displaymath}
\begin{split}
\mathcal{E}^\nu_H(\Tilde{v}_{H}) &= \int_{\Gamma_\infty} \lVert\nabla \Tilde{v}_{H}\rVert^2 + \int_{\Gamma_\infty}\frac{5}{3}\Tilde{u}_{1,H}^{4/3}\Tilde{v}_{H}^2 +  \int_{\Gamma_\infty}\psi_{1,H} \Tilde{v}_{H}^2 \\ 
 &+\int_{\Gamma_\infty} (\Tilde{u}_{1,H}+\Tilde{v}_{H})^{10/3}-\Tilde{u}_{1,H}^{10/3} - \frac{5}{3}\Tilde{u}_{1,H}^{4/3}\Tilde{v}_{H}^2 - \frac{10}{3}\Tilde{u}_{1,H}^{7/3}\Tilde{v}_{H} \\ 
 &+ \frac{1}{2}\int_{\Gamma_\infty} \psi_H(\Tilde{v}_{H}^2 + 2u_1\Tilde{v}_{H} -\nu) - \int_{\Gamma_\infty} \psi_{1,H} \nu.
\end{split}
\end{displaymath}
In the beginning of this section we showed that $\Tilde{u}_{1,H}, \Tilde{v}_{H}, \psi_{1,H}, \psi_{H}$ converge everywhere along a subsequence to $u_1, v, \phi_1, \phi$ pointwise almost everywhere, respectively. Note that in this step whenever we are taking a limit or liminf, it is on the convergent subsequence of the thin film solutions. For convenience, we still represent this subsequence by the subscript $H$. Using Fatou's lemma and lemma~\ref{lem:KEhelp}, we see that
\begin{equation}
\label{eq:liminf_local1}
\begin{split}
    \int_{\Gamma_\infty} (\Tilde{u}_{1,H}+v)^{10/3}-\Tilde{u}_{1,H}^{10/3} - \frac{5}{3}\Tilde{u}_{1,H}^{4/3}v^2 - \frac{10}{3}\Tilde{u}_{1,H}^{7/3}v  &\leq \liminf_{H \rightarrow \infty} \left(\int_{\Gamma_\infty} (\Tilde{u}_{1,H}+\Tilde{v}_{H})^{10/3}-\Tilde{u}_{1,H}^{10/3}\right. \\  
    & - \left.\int_{\Gamma_\infty} \frac{5}{3}\Tilde{u}_{1,H}^{4/3}\Tilde{v}_{H}^2 + \frac{10}{3}\Tilde{u}_{1,H}^{7/3}\Tilde{v}_{H}\right).
\end{split}
\end{equation}
 Further, $\psi_{1,H}$ is uniformly bounded almost everywhere, we get, using the compactness of $\nu$, that 
\begin{equation}
    \label{eq:liminf_local4}
    \lim_{H \rightarrow \infty} \int_{\Gamma_\infty} \psi_{1,H} \nu = \int_{\Gamma_\infty} \phi_{1} \nu.
\end{equation}

We also have that $\psi_H$ converges to $\phi$ pointwise almost everywhere along a subsequence. Hence using the uniform bound of $\psi_H$ in $H^2_{\text{unif}}(\mathbb{R}^3)$, we conclude that $\nabla \psi_H$ converges to $\nabla \phi$ pointwise almost everywhere, along a subsequence, by uniqueness of weak limits. We again represent this subsequence by the subscript $H$, for convenience.    

Further using lemma~\ref{lem:KEhelp} and lemma~\ref{lem:Lposdef} we see that for large enough $H$, $\int_{\Gamma_\infty} |\nabla \psi_H|^2$ is bounded by $2\mathcal{E^\nu}(0) + 1$. Hence using Fatou's lemma,
\begin{equation}
\label{eq:liminf_local2}
\begin{split}
    \frac{1}{2}\int_{\Gamma_\infty} \phi(v^2 + 2u_1v -\nu) &= 
    \frac{1}{2}\int_{\Gamma_\infty} \lVert\nabla \phi\rVert^2\\
    &\leq \liminf_{H \rightarrow \infty}  \frac{1}{2}\int_{\Gamma_\infty} \lVert\nabla \psi_H\rVert^2\\
    &= \liminf_{H \rightarrow \infty} \frac{1}{2}\int_{\Gamma_\infty} \psi_H(\Tilde{v}_{H}^2 + 2u_1\Tilde{v}_{H} -\nu).
\end{split}
\end{equation}
Note that we have used the fact that 
$\nabla^2 \psi_H = \Tilde{v}_{H}^2 +2u_1\Tilde{v}_{H} -\nu$ and $\nabla \psi_H \in (L^2_{per}(\Gamma_\infty))^3$. Using Lemma ~\ref{lem:Lposdef} we get
\begin{displaymath}
\begin{split}
    \int_{\Gamma_\infty} \lVert\nabla \Tilde{v}_{H}\rVert^2 + \int_{\Gamma_\infty}f_H\Tilde{v}_{H}^2  \geq &\int_{\Gamma_\infty} \lVert\nabla v\rVert^2 + \int_{\Gamma_\infty}f_Hv^2 \\&- 2\int_{\Gamma_\infty} (\Tilde{v}_{H}-v)\nabla^2 v + 2\int_{\Gamma_\infty}f_H(\Tilde{v}_{H}-v)v,
\end{split}
\end{displaymath}
where $f_H = \frac{5}{3}\Tilde{u}_{1,H}^{4/3} + \psi_{1,H}$. Note that $f_H$ and $\Tilde{v}_{H}-v$ are uniformly bounded almost everywhere. Therefore, as $H \rightarrow \infty$, $\int_{\Gamma_\infty} (\Tilde{v}_{H}-v) \nabla^2 v$ and $\int_{\Gamma_\infty}f_H(\Tilde{v}_{H}-v)v$ converge to $0$ using DCT as $v \in H^2(\Gamma_\infty) \cap L^1(\Gamma_\infty)$. Hence,
\begin{equation*}
    \liminf_{H \rightarrow \infty} \int_{\Gamma_\infty} \lVert\nabla v\rVert^2 + \int_{\Gamma_\infty}f_H v^2 \leq \liminf_{H \rightarrow \infty} \int_{\Gamma_\infty} \lVert\nabla \Tilde{v}_{H}\rVert^2 + \int_{\Gamma_\infty}f_H \Tilde{v}_{H}^2 .
\end{equation*}
Since $v \in L^1\cap L^2(\Gamma_\infty)$ and $f_H$ is uniformly bounded almost everywhere and converge to $\frac{5}{3}u_{1}^{4/3} + \phi_{1}$ pointwise almost everywhere, using DCT we obtain 
\begin{equation}
 \label{eq:liminf_local3}
        \int_{\Gamma_\infty} \lVert\nabla v\rVert^2 + \int_{\Gamma_\infty}\frac{5}{3}u_{1}^{4/3}v +  \int_{\Gamma_\infty}\psi_{1} v^2 \leq \liminf_{H \rightarrow \infty} \int_{\Gamma_\infty} \lVert\nabla \Tilde{v}_{H}\rVert^2 + \int_{\Gamma_\infty}\Big(\frac{5}{3}\Tilde{u}_{1,H}^{4/3} + \psi_{1,H}\Big)\Tilde{v}_{H}^2 
\end{equation}
This concludes the proof of the theorem.
\end{proof}

\subsection{Uniqueness of the minimizer}
Since $w \in \mathcal{Q}_1^\nu$ implies that $w \in L^1\cap L^2(\Gamma_\infty)$, we get, using equation~\eqref{eq:TFW_smeared_EL}, that
\begin{displaymath}
\begin{split}
 \mathcal{E}^\nu(w) &= \int_{\Gamma_\infty} (\lVert\nabla (u_1+w)\rVert^2 - \lVert\nabla u_1\rVert^2) + \int_{\Gamma_\infty}(w^2+2u_1w) \phi_1  \int_{\Gamma_\infty} \left( |u_1 + w|^{10/3} - u_1^{10/3} \right) \\ 
 &+ \frac{1}{2}D_G(w^2 + 2u_1w -\nu, w^2 + 2u_1w -\nu)  - \int_{\Gamma_\infty} \phi_1 \nu.
\end{split}
\end{displaymath}
Using an argument similar to that presented in section 4.4 of \cite{CE11}, we see that $v$ is minimizer of $\mathcal{E}^\nu_1$ over $\mathcal{Q}_1^\nu$ if and only if $(u_1+v)^2$ solves the minimization problem  
$$\inf_{\rho\in\mathcal{K}} \; G(\rho),$$
where 
\begin{displaymath}
\begin{split}
G(\rho) &= \int_{\Gamma_\infty} (\lVert\nabla \sqrt{\rho}\rVert^2 - \lVert\nabla u_1\rVert^2) +  \int_{\Gamma_\infty}(\rho - u_1^2) \phi_1 + \int_{\Gamma_\infty} \left( \rho^{5/3} - u_1^{10/3} \right)\\
& + \frac{1}{2}D_G(\rho - u_1^2 -\nu, \rho - u_1^2 -\nu) , \\
\mathcal{K} = \Big\{ & \rho \geq 0 \;:\; \sqrt{\rho} - u_1 \in H^1_{per}(\Gamma_\infty), |x_3|(\rho - u_1^2) \in L^1(\Gamma_\infty),\\
& |x_3|(\sqrt{\rho}-u_1)^2 \in L^1(\Gamma_\infty), \int_{\Gamma_\infty}(\rho - u_1^2)=\int_{\Gamma_\infty}\nu \Big\}. 
\end{split}
\end{displaymath}
It is straightforward to see that $\mathcal{K}$ is a convex set and that $D_G(\rho - u_1^2 - \nu, \rho - u_1^2 - \nu)$ is convex on $\mathcal{K}$, as shown by Proposition 2.3 in~\cite{BTF00}. Therefore, $G(\rho)$ is strictly convex over $\mathcal{K}$, implying that $\mathcal{G}$ can have at most one minimizer within $\mathcal{K}$. We thus see that $v$ is the unique minimizer of $\mathcal{E}^\nu_1$ over $\mathcal{Q}_1^\nu$.  

\section{Discussion and Conclusion}
We present a few applications of the main results of the paper and discuss certain outstanding questions that remain unresolved. 

\subsection{Approximate evaluation of relative energies of quasi-planar defects}
A central conclusion of the results presented here is the \emph{short range} nature of the electronic response of the TFW model to nuclear perturbations, even for semi-infinite configurations like the one shown in Figure~\ref{fig:planar_defect_schematic}. It is to be noted that in practice, we rarely evaluate energies over infinite domains, but rather compute them on finite domains large enough for the electronic fields due to the perturbation to decay. The theory developed in this work helps us quantify the convergence of the corresponding approximate relative energy to the true relative energy. Suppose that the reference nuclear distribution $m_1$ is smooth. The perturbed nuclear distribution is denoted as $m_2 = m_1 + \nu$, with the perturbation $\nu$ as in equation~\eqref{eq:nearly_planar_perturbation}. Let us define the approximate relative energy $\tilde\gamma^\nu_{K}$, with $K \in \mathbb{N}$, as
   \begin{displaymath}
       \begin{split}
           \tilde\gamma_{K}^\nu &= \int_{\Gamma_K} \lVert\nabla v\rVert^2 + \int_{\Gamma_K} (u_1 + v)^{10/3}- u_{1}^{10/3} - \frac{10}{3}u_1^{7/3}v\\  
 &+ \frac{1}{2}\int_{\Gamma_K} \phi ((u_1 + v)^2 - u_1^2 -\nu) + \int_{\Gamma_K} \phi_{1} (v^2  - \nu).\\
       \end{split}
   \end{displaymath} 
The difference between the approximate relative energy $\tilde\gamma^\nu_{K}$ and the true relative energy $\gamma^\nu_1$ can be controlled, as shown by the following proposition:
\begin{proposition}
If $K>L_0$ there exist positive constant $C_1$ and $C_2$ such that,
   \begin{equation*}
       \lVert \gamma^\nu_1 - \tilde\gamma_K^\nu\rVert \le C_1 \exp(-C_2(K - L_0)).
   \end{equation*}
\end{proposition}

\begin{proof}
We have, according to lemma~\ref{lem:KEhelp}, the existence of a positive constant $C$ such that
\begin{equation}
\label{eq:app_algo_loc1}
     \left \lvert \int_{\Gamma_\infty \setminus \Gamma_K} (u_1 + v)^{10/3}- u_{1}^{10/3} - \frac{10}{3}u_1^{7/3}v \right \rvert \leq C \int_{\Gamma_\infty \setminus \Gamma_K} (v^2 + |v|^{10/3}).
\end{equation}
Using equation~\eqref{eq:app_algo_loc1} and the compactness of $\nu$ in $x_3$ direction we see that 
\begin{displaymath}
    \begin{split}
        \lVert \gamma^\nu_1 - \tilde\gamma_K^\nu\rVert &\leq \int_{\Gamma_\infty \setminus \Gamma_K} \lVert\nabla v\rVert^2 + C \int_{\Gamma_\infty \setminus \Gamma_K} (v^2 + |v|^{10/3})\\ 
        & +  \frac{1}{2}\int_{\Gamma_\infty \setminus \Gamma_K} |\phi (v^2+2u_1v)| + \int_{\Gamma_\infty \setminus \Gamma_K} |\phi_{1}| v^2.\\
    \end{split}
\end{displaymath}
According to Lemma~\ref{thm:density_decay},
$\|\nabla v(y)\|$, $|v(y)|$, and $|\phi(y)|$, for any $y \in \Gamma_\infty \setminus \Gamma_{K}$, are bounded by 
$k_1e^{-k_2 dist(y, \Gamma_{L_0})}$ for positive constants $k_1, k_2$ which depend only on the uniform bounds of $m_1$ whenever $K > L_0$.
Using this along with the fact that $u_1$ and $\phi_1$ belong to $L^\infty(\mathbb{R}^3)$ yields the desired bounds.
\end{proof}

\subsection{Stacking fault energy in bicrystals}

As alluded to in the introduction, an important application of the theory developed here is the study of generalized stacking fault energies. The simplest approximation in this regard is obtained by sliding one half of a perfect crystal over the other half along a specified direction on the plane separating the two halves. The relative energy, called the \emph{generalized stacking fault energy}, plays a key role in the study of dislocations in crystals. As an application of the theory developed in this work, we study stacking faults by a slight modification of the geometry considered earlier. Let $m_1 \in \mathcal{M}$ be a nuclear distribution that is periodic, with unit period, and symmetric about each of the coordinate axis. Define, for $L_0 \in \mathbb{N}$, the domain $\Tilde{\Gamma}_{L_0} = (-\frac{1}{2}, \frac{1}{2}) \times (-\frac{1}{2}, \frac{1}{2}) \times (-\frac{1}{2}, L_0 - \frac{1}{2})$. To study the generalized stacking fault energy, we shift the nuclear distribution contained in $\tilde\Gamma_{L_0}$ by a distance $b$ along the $x_1$ direction; a similar analysis can be carried out for a shift along any other direction. The corresponding nuclear perturbation $\nu$ is then given by
$$\nu|_{\Gamma_\infty}(x)=\begin{cases}
m_b(x) - m_1(x), & x \in \overline{\Tilde{\Gamma}}_{L_0},\\ 
0, & x \in \Gamma_\infty \setminus \Tilde{\Gamma}_{L_0} .
\end{cases} $$
where $m_b$ is defined as $m_b(x_1, x_2, x_3) = m_1(x_1 + b, x_2, x_3)$. To get bounds on the generalized stacking fault energy, we define two new fields $\tilde{v}$ and $\tilde\phi$ as follows: 
\begin{displaymath}
\begin{split}
\Tilde{v}(x) &= \begin{cases}
v(x) + u_1(x) - u_b(x), & x \in \overline{\Tilde{\Gamma}}_{L_0},\\ 
v(x), & x \in \Gamma_\infty \setminus \Tilde{\Gamma}_{L_0}
\end{cases}\\
\Tilde{\phi}(x) &= \begin{cases}
\phi(x) + \phi_1(x) - \phi_b(x), & x \in \overline{\Tilde{\Gamma}}_{L_0},\\ 
\phi(x), & x \in \Gamma_\infty \setminus \Tilde{\Gamma}_{L_0} .
\end{cases}
\end{split}
\end{displaymath}
In the equations displayed above, $u_b(x) = u_1(x_1 + b, x_2, x_3)$, and $\phi_b(x) = \phi_1(x_1 + b, x_2, x_3)$. Expressing $v$ and $\phi$ in terms of these fields, we see that 
\begin{align*}
    \mathcal{E}^\nu_1(v) &= \int_{\Gamma_\infty} \|\nabla \Tilde{v}\|^2 + 2\int_{\Tilde{\Gamma}_{L_0}} \Big( \|\nabla u_1\|^2 - \nabla u_b \cdot \nabla u_1 - \nabla \Tilde{v}\cdot \nabla u_1 + \nabla \Tilde{v} \cdot \nabla u_b  \Big)\\
    &+ \int_{\Gamma_\infty \setminus \Tilde{\Gamma}_{L_0}} \Big(|u_1 + \Tilde{v}|^{10/3} - u_1^{10/3} - \frac{10}{3}u_1^{7/3}\Tilde{v} \Big)\\
    &+ \int_{\Tilde{\Gamma}_{L_0}}\Big(|u_b + \Tilde{v}|^{10/3} - u_1^{10/3} - \frac{10}{3}u_1^{7/3}\Tilde{v} \Big) + \frac{10}{3}\int_{\Tilde{\Gamma}_{L_0}} (u_1^{10/3} - u_1^{7/3}u_b) \\
    &+ \int_{\Gamma_\infty} \phi_1 \Tilde{v}^2 -\int_{\Tilde{\Gamma}_{L_0}} \phi_1\nu + \int_{\Tilde{\Gamma}_{L_0}} \phi_1(2\Tilde{v}u_1 - 2\Tilde{v}u_b  ) + \int_{\Tilde{\Gamma}_{L_0}} \phi_1(u_b - u_1)^2 \\
    & + \frac{1}{2}\int_{\Gamma_\infty} \Tilde{\phi}\Tilde{v}^2 + \frac{1}{2}\int_{\Tilde{\Gamma}_{L_0}} \Tilde{\phi} \Big( (u_b-u_1)^2 + 2\Tilde{v}(u_b-u_1) \Big)
     + \frac{1}{2}\int_{\Tilde{\Gamma}_{L_0}} \Big(\Tilde{v}^2 + 2 \Tilde{v}(u_b-u_1)\Big)(\phi_b - \phi_1)\\ 
    &+ \frac{1}{2}\int_{\Tilde{\Gamma}_{L_0}} (\phi_b-\phi_1) (u_b-u_1)^2 + \int_{\Gamma_\infty} u_1\Tilde{\phi}\Tilde{v} + \int_{\Tilde{\Gamma}_{L_0}} u_1 \Big(\Tilde{\phi}(u_b-u_1) + \Tilde{v} (\phi_b - \phi_1)\Big) \\
    & + \int_{\Tilde{\Gamma}_{L_0}} u_1 (u_b-u_1)(\phi_b-\phi_1) - \frac{1}{2}\int_{\Tilde{\Gamma}_{L_0}} \nu(\Tilde{\phi} - \phi_1 +\phi_b).
\end{align*}
Weakly enforcing equation~\eqref{eq:TFW_smeared_EL} for the periodic crystal with test functions $u_1$ and $u_b$ over $\Gamma_{L_0}$ yields 
\begin{align*}
    \mathcal{E}^\nu_1(v) &= \int_{\Gamma_\infty} \|\nabla \Tilde{v}\|^2 + 2\int_{\Tilde{\Gamma}_{L_0}} \Big( - \nabla \Tilde{v}\cdot \nabla u_1 + \nabla \Tilde{v} \cdot \nabla u_b  \Big)\\
    &+ \int_{\Gamma_\infty \setminus \Tilde{\Gamma}_{L_0}} \Big(|u_1 + \Tilde{v}|^{10/3} - u_1^{10/3} - \frac{10}{3}u_1^{7/3}\Tilde{v} \Big)\\
    &+ \int_{\Tilde{\Gamma}_{L_0}}\Big(|u_b + \Tilde{v}|^{10/3} - u_1^{10/3} - \frac{10}{3}u_1^{7/3}\Tilde{v} \Big)  \\
    &+ \int_{\Gamma_\infty} \phi_1 \Tilde{v}^2 -\int_{\Tilde{\Gamma}_{L_0}} \phi_1\nu + \int_{\Tilde{\Gamma}_{L_0}} \phi_1(2\Tilde{v}u_1 - 2\Tilde{v}u_b  ) + \int_{\Tilde{\Gamma}_{L_0}} \phi_1(u_b^2 - u_1^2) \\
    & + \frac{1}{2}\int_{\Gamma_\infty} \Tilde{\phi}\Tilde{v}^2 + \frac{1}{2}\int_{\Tilde{\Gamma}_{L_0}} \Tilde{\phi} \Big( (u_b-u_1)^2 + 2\Tilde{v}(u_b-u_1) \Big)
     + \frac{1}{2}\int_{\Tilde{\Gamma}_{L_0}} \Big(\Tilde{v}^2 + 2 \Tilde{v}(u_b-u_1)\Big)(\phi_b - \phi_1)\\ 
    &+ \frac{1}{2}\int_{\Tilde{\Gamma}_{L_0}} (\phi_b-\phi_1) (u_b-u_1)^2 + \int_{\Gamma_\infty} u_1\Tilde{\phi}\Tilde{v} + \int_{\Tilde{\Gamma}_{L_0}} u_1 \Big(\Tilde{\phi}(u_b-u_1) + \Tilde{v} (\phi_b - \phi_1)\Big) \\
    & + \int_{\Tilde{\Gamma}_{L_0}} u_1 (u_b-u_1)(\phi_b-\phi_1) - \frac{1}{2}\int_{\Tilde{\Gamma}_{L_0}} \nu(\Tilde{\phi} - \phi_1 +\phi_b).
\end{align*}
Note that we have used the fact that surface terms vanish as $\nabla u \cdot n$ is zero on the boundary due to symmetry and periodicity. Using lemma~\ref{lem:KEhelp}, we get
\begin{align*}
    \mathcal{E}^\nu_1(v) &= \int_{\Gamma_\infty} \|\nabla \Tilde{v}\|^2 + 2\int_{\Tilde{\Gamma}_{L_0}} \Big( - \nabla \Tilde{v}\cdot \nabla u_1 + \nabla \Tilde{v} \cdot \nabla u_b  \Big)\\
    &+  k_1 \int_{\Gamma_\infty} \Big( |\Tilde{v}| + |\Tilde{v}|^{10/3} + \Tilde{v}^2 \Big)  \\
    &+ \int_{\Gamma_\infty} \phi_1 \Tilde{v}^2 -\int_{\Tilde{\Gamma}_{L_0}} \phi_1\nu + \int_{\Tilde{\Gamma}_{L_0}} \phi_1(2\Tilde{v}u_1 - 2\Tilde{v}u_b  ) + \int_{\Tilde{\Gamma}_{L_0}} \phi_1(u_b^2 - u_1^2) \\
    & + \frac{1}{2}\int_{\Gamma_\infty} \Tilde{\phi}\Tilde{v}^2 + \frac{1}{2}\int_{\Tilde{\Gamma}_{L_0}} \Tilde{\phi} \Big( (u_b-u_1)^2 + 2\Tilde{v}(u_b-u_1) \Big)
     + \frac{1}{2}\int_{\Tilde{\Gamma}_{L_0}} \Big(\Tilde{v}^2 + 2 \Tilde{v}(u_b-u_1)\Big)(\phi_b - \phi_1)\\ 
    &+ \frac{1}{2}\int_{\Tilde{\Gamma}_{L_0}} (\phi_b-\phi_1) (u_b-u_1)^2 + \int_{\Gamma_\infty} u_1\Tilde{\phi}\Tilde{v} + \int_{\Tilde{\Gamma}_{L_0}} u_1 \Big(\Tilde{\phi}(u_b-u_1) + \Tilde{v} (\phi_b - \phi_1)\Big) \\
    & + \int_{\Tilde{\Gamma}_{L_0}} u_1 (u_b-u_1)(\phi_b-\phi_1) - \frac{1}{2}\int_{\Tilde{\Gamma}_{L_0}} \nu(\Tilde{\phi} - \phi_1 +\phi_b).
\end{align*}
for some some positive constant $k_1$. Collecting all the terms not involving any function of $\Tilde{v}$ and $\Tilde{\phi}$ and using the fact that $u_1, u_b, \phi_1, \phi_b, \Tilde{v}, \Tilde{\phi}$ all belong to $L^\infty(\mathbb{R}^3)$, we get
\begin{displaymath}
    \begin{split}
          |\mathcal{E}^\nu(v)| &\leq \Tilde{C}( \|\Tilde{\phi}\|_{L^1(\Gamma_\infty)} + \|\Tilde{v}\|_{L^1(\Gamma_\infty)} +  \|\Tilde{v}\|_{H^1(\Gamma_\infty)}^2 \\
    & + \frac{1}{2}\Big|\int_{\Tilde{\Gamma}_{L_0}} (\phi_b + \phi_1)(u_b^2-u_1^2 -\nu)\Big|,   
    \end{split}
\end{displaymath}
where $\Tilde{C}$ is a positive constant only dependent on uniform bounds of $m_1$.

Note that $\int_{\Gamma_0}(\phi_b u_b^2 - \phi_1 u_1^2 ) = 0$ due to the periodicity of $u_1$ and $\phi_1$. We also have $\int_{\Gamma_0}(\phi_1 u_b^2 - \phi_b u_1^2 ) = 0$ due to the periodicity and symmetry of $u_1$ and $\phi_1$. Therefore, 
$$\frac{1}{2}\left\lvert\int_{\Tilde{\Gamma}_{L_0}} (\phi_b + \phi_1)(u_b^2-u_1^2 -\nu)\right\rvert = \frac{1}{2}\left\lvert\int_{\Tilde{\Gamma}_{L_0}} (\phi_b + \phi_1)\nu\right\rvert.$$
A similar argument employing the symmetry and periodicity of $\phi_1$ shows that the term above on the right hand side is zero. Using proposition~\ref{prop:ortner_TFW_stability_self} to compare $m_2$ and $m_b$ in $\Tilde{\Gamma}_{L_0}$, and $m_2$ and $m_1$ in $\Gamma_\infty \setminus \Tilde{\Gamma}_{L_0}$, we get 
$$\Tilde{C}(\|\Tilde{\phi}\|_{L^1(\Gamma_\infty)} + \|\Tilde{v}\|_{L^1(\Gamma_\infty)} + \|\Tilde{v}\|_{H^1(\Gamma_\infty)}) \leq C_1 + C_2e^{-C_3L_0},$$
where $C_1, C_2, C_3$ are positive constants only dependent on uniform bounds of $m_1$. We therefore obtain
\begin{displaymath}
|\mathcal{E}^\nu_1(v)| \le C_1 + C_2e^{-C_3L_0}.
\end{displaymath}
In the limit $L_0 \to \infty$, we thus see that the constant term $C_1$ provides the upper bound on the desired stacking fault energy. Note also that if $m_1|_{\Gamma_0}$ is smooth then we can also say $I^\nu = \mathcal{E}^\nu(v)$, thanks to Theorem~\ref{thm:reflexivemin}. We have thus shown that the surface energy of a bicrystal, with symmetric and periodic nuclear density, is finite in the TFW setting. We were, however, not able to pose a minimization problem for this case.

This analysis thus guarantees that we can obtain a reasonably accurate stacking fault energy by considering finite bicrystal domains since the stacking fault energy converges exponentially as a function of the domain size $L_0$. This also justifies the additivity of the bulk and surface energies that is often assumed as the starting point for many numerical analyses. While this is numerically evident, the arguments presented here provide a mathematical justification for such calculations in the context of TFW models.

\subsection{Extension of the TFW model with Dirac Exchange}

As mentioned in the introduction, the TFW model is a simple but non-trivial model for orbital free DFT. A more accurate model, albeit still insufficient for predictive calculations, called the Thomas-Fermi-von Weizs\"acker-Dirac (TFWD) model, includes an additional term to account for exchange energy. Focusing on a nuclear distribution of finite extent, we can write the energy of the TFWD model as 
\begin{displaymath}
E^{TFWD}_\Lambda(\rho) = E_\Lambda(\rho) - C_D \int_{\mathbb{R}^3} \rho^{4/3}(x)\,dx.
\end{displaymath}
In the equation displayed above, $C_D =  \frac{3}{4}(\frac{3}{\pi})^{\frac{1}{3}}$. It is evident that the energy in the TFWD model is non-convex, and hence the thermodynamic limit analysis presented in this work is not directly applicable. While previous analytical studies such as \cite{Ricaud18} have examined the effect of non-convexity by varying the parameter $C_D$, they do not provide quantitative bounds on $C_D$ that ensure a unique density, which is crucial for establishing the locality results proven in the TFW model.

We can, however, solve the TFWD model numerically for a given nuclear distribution. For simplicity, we focus on a homogeneous defect to investigate the effect of the Dirac term on the locality of the fields. This choice is further supported by the findings in \cite{GLM21}, which showed that, for homogeneous 2D materials, the electronic fields obtained from the Thomas–Fermi model closely approximate those obtained from the more general reduced Hartree–Fock model. We use a staggered algorithm that uses an augmented Lagrangian approach to enforce the charge constraint, as outlined in Appendix~\ref{app:tfw_numerics}. We use $C_W=1$ and $C_{TF} = 3.2$ for all the results shown here. As a simple illustration, we consider a nuclear distribution $m_1$ that is a mollified Dirac comb:
$$
m_1(z) = \sum_{k \in \mathbb{R}} \frac{1}{\sqrt{2\pi\sigma^2}}\exp\left(-\left(\frac{(z - k)}{\sigma}\right)^2\right).
$$
The width of the nuclear distribution $\sigma$ is chosen such that it is much smaller than then periodicity of the crystal: $\sigma \ll 1$. We consider a perturbation $\nu$ of the form 
\begin{displaymath}
\nu(z) = \frac{M}{\sqrt{2\pi\sigma^2}}\exp\left(-\frac{z^2}{\sigma^2}\right),
\end{displaymath}
for some $M > 0$. The reference and perturbed crystal are shown in Figure~\ref{fig:gaussian_density}. Physically the reference nuclear configurations correspond to a set of thin slabs that are periodically spaced along the $z$ axis with unit period, and the perturbed nuclear configurations corresponds to the case where one of the slabs has a different nuclear charge. 
\begin{figure}[h]
    \centering
    \includegraphics[width=0.5\linewidth]{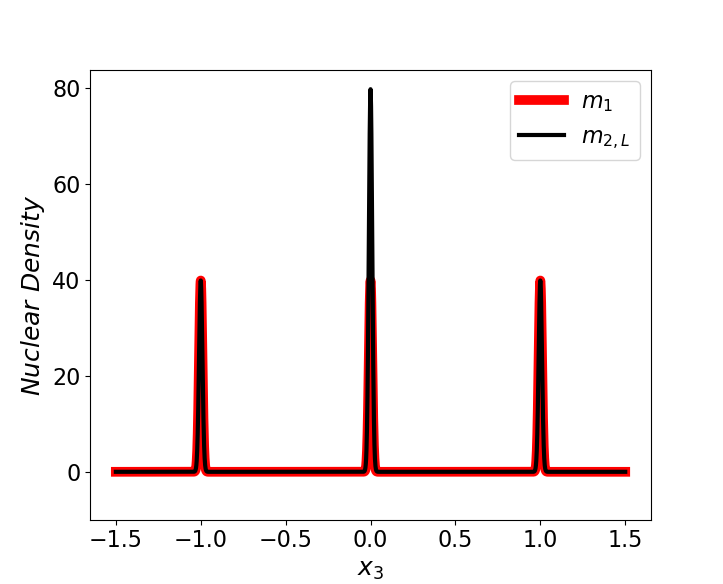}
    \caption{Perfect crystal modeled by a smooth Dirac comb and perturbed crystal with perturbations on one of the comb's fingers.}
    \label{fig:gaussian_density}
\end{figure}

We show in Figure~\ref{fig:tfw_vs_tfwd_smooth_dirac} the numerical solutions obtained using the TFW and TFWD models.
\begin{figure}[h]
\begin{subfigure}{0.45\textwidth}
\centering
\includegraphics[width=0.9\textwidth]{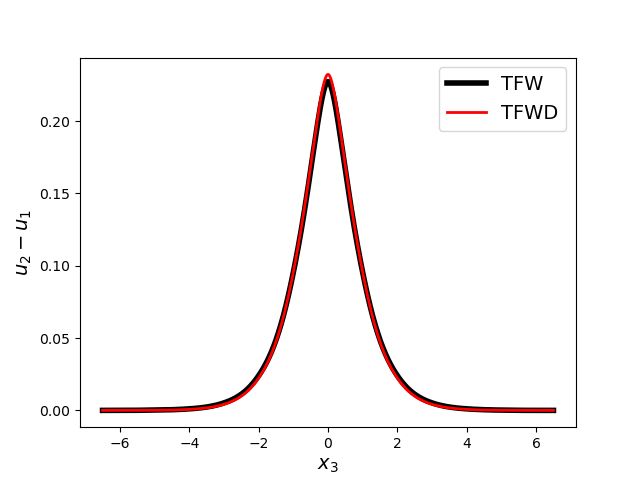}
\label{Difference in square root electronic density}
\end{subfigure}
~
\begin{subfigure}{0.45\textwidth}
\centering
\includegraphics[width=0.9\textwidth]{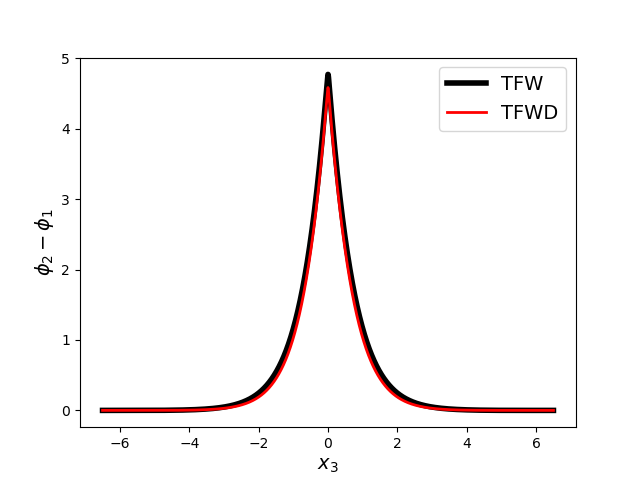}
\label{Difference in potential}
\end{subfigure}
\caption{Comparison of TFW and TFWD models for the nuclear distributions shown in Figure~\ref{fig:tfw_vs_tfwd_smooth_dirac}. Both the difference in square root density and the difference in potential are shown. The results indicate that the TFWD response shares the decay properties of the TFW model, despite being non-convex.}
\label{fig:tfw_vs_tfwd_smooth_dirac}
\end{figure}
Periodic boundary conditions were used for both the reference and perturbed nuclear configurations; more details are presented in Appendix~\ref{app:tfw_numerics}. The numerical results indicate that the electronic fields in the TFWD model are very close to those predicted by the TFW model. In particular, this indicates that the locality of the TFW model extends to the TFWD model too. We however do not have a rigorous proof of this statement since the TFWD energy functional is non-convex, and the analytical methods used here are thus not directly applicable. The numerical similarity does, however, raise the interesting question as to whether a proof of the thermodynamic limit that doesn't use convexity of the energy functional in an essential way is possible. We relegate this for a future study. 

\subsection{Conclusion}
The present work extends previous studies on the thermodynamic limit of the energy of point defects within the framework of the Thomas-Fermi-von Weizs\"acker model, as presented in \cite{CE11}, to an analysis of quasi-planar defects. We analyze nuclear perturbations that are restricted to a thin slab in an otherwise perfect crystal and prove the convergence of the electron density and potential for these class of perturbations. The key contribution of this paper is the formulation of a well-posed minimization problem to compute the relative energy of these perturbed configurations. In particular, we show the difference in the square root electron densities between the defective and perfect crystals is a unique minimizer of a relative energy functional, which can also be evaluated by solving a Schr\"odinger-Poisson system of PDEs. As an application of these results, we show that the stacking fault energy for unidirectional slip is finite within the TFW setting. This analysis also formally establishes the additivity hypothesis, in the context of TFW models, for the bulk and surface energies that is usually assumed in continuum models like the Peierls-Nabarro model.  

There are several outstanding questions that remain: (i) The class of perturbations considered in this work could potentially be enlarged by removing the compactness condition, thus extending the theory to a larger class of defective nuclear perturbations. (ii) The theoretical framework developed in this work, particularly those related to the generalized stacking fault energy, could be used in conjunction with the Peierls-Nabarro model to infer useful bounds on the energetics of dislocation core structures. (iii) The preliminary numerical results presented in this work indicate that many of the locality and decay properties of the TFW model also extend to the TFWD model, but there is no straightforward means to prove this rigorously as the TFWD functional is non-convex. Density functionals used in orbital-free DFT are non-convex in general, and the development of the right analytical tools to study defects using these more general DFT models remains an open problem. (iv) Though the emphasis of this work is entirely analytical, the eventual goal of this line of research is to identify simplified variational formulations for defect energies that can be used to develop faster numerical algorithms. These and related issues will be explored in future works. 

\section{Supplementary Results}
We discuss first a simple extension of the thermodynamic limit argument for thin films to relax a symmetry requirement. We then list proofs of several propositions that are used in the theorems above. 

\subsection{Modification of thin film solution}
\label{subsec:thinfilmmod}
 Given a positive bounded nuclear density which is compact in $x_2, x_3$ direction, but of infinite extent in the $x_1$ direction, let us denote the polymer nuclear density with width $\Lambda$ in $x_2$ direction and height $H$ in $x_3$ direction as $m_{H, \Lambda}$. From subsection~\ref{subsec:polymers}, there exists a unique solution $(u_{H, \Lambda}, \phi_{H, \Lambda})$ in $(L^2_{\text{unif}} \cap L^{7/3}_{\text{loc}}(\mathbb{R}^3)) \times L^1(\mathbb{R})$ satisfying the PDE:
\begin{equation} \label{eq:ELrod}
\begin{split}
    -\nabla^2  u_{H, \Lambda} &+ \frac{5}{3}u_{H, \Lambda}^{7/3} + u_{H, \Lambda} \phi_{H, \Lambda} = 0, \\
-\nabla^2 \phi_{H, \Lambda} &= 4\pi (u_{H, \Lambda}^2 - m_{H, \Lambda}).
\end{split}
\end{equation}
Moreover, $u_{H, \Lambda}, \phi_{H, \Lambda}$ are periodic in $x_1$ and  $u_{H, \Lambda}, \phi_{H, \Lambda}$ are uniformly bounded (independent of $\Lambda$ or $H$) in $H^4_{\text{unif}}(\mathbb{R}^3) \times H^2_{\text{unif}}(\mathbb{R}^3)$.
Further note that $u_H$ is uniformly bounded almost everywhere and for $x$ outside support of $m_{\Lambda, H}$, 
$|u_{H, \Lambda}(x)|+ |\nabla u_{H, \Lambda}(x)| \leq \frac{C}{(x_2^2 + x_3^2)^{3/2}}$, where $C$ is independent of $\Lambda$ and $H$. 

Further, assume that nuclear density is symmetric in $x_2$ direction. Using the uniform bounds of $u_{H, \Lambda}, \phi_{H, \Lambda}$ in $H^4_{\text{unif}}(\mathbb{R}^3) \times H^2_{\text{unif}}(\mathbb{R}^3)$ we take the limit as $\Lambda$ goes to infinity by ensuring symmetry and periodicity in of $m_{H, \Lambda}$ in $x_2$ direction. Here, $m_{H, \Lambda}$ converges to the same $m_H$ defined in the subsection~\ref{subsec:thinfilm} section which corresponds to a thin film without the restriction of symmetry in $x_1$ direction. Let's call the converged solutions, along a subsequence, as $(u_H, \phi_H)$, and passing as the local limit in \eqref{eq:ELrod}
\begin{equation}
\label{eq:ELthinfilm2}
    \begin{split}
         &-\nabla^2  u_{H} + \frac{5}{3}u_{H}^{7/3} + u_{H} \phi_{H} = 0, \\
&-\nabla^2 \phi_{H} = 4\pi (u_{H}^2 - m_{H}),\\
&u_H \geq 0.   
    \end{split}  
\end{equation}

This also implies that $(u_H, \phi_H)$ is uniformly bounded in $H^4_{\text{unif}}(\mathbb{R}^3) \times H^2_{\text{unif}}(\mathbb{R}^3)$ and is periodic in $x_1$. Moreover, for $x$ outside support of $m_H$, $|u_{H}(x)|+ |\nabla u_{H}(x)| \leq \frac{C}{x_3^{3/2}}$, where $C$ is independent of $H$. 

It is easy to deduce that $u_H \in H^1(\Gamma_\infty)$. Now use the periodicity of $u_H$ in $x_1$ direction and symmetry in $x_2$ direction to conclude that $u_H$ is the minimizer of equation~\ref{eq:thinfilmmin} corresponding to the new $m_H$ which is not symmetric in $x_1$ direction. This follows from exact steps in Theorem 3.2 ~\cite{BTF00}. Moreover, $u_H$ is uniformly bounded almost everywhere independent of $H$ and for any $R> 0$ there exists a positive constant $\nu_R$ such that $\inf_{|x_3|< R} u_H \geq \nu_R$. The later conclusion follows from Proposition 3.5 ~\cite{BTF00}. In addition, if we assume $m_H$ to be smooth
then $\phi_H$ is also uniformly bounded almost everywhere.

\subsection{Supplementary Lemmas}
\label{app_2}
\begin{lemma} 
\label{lem:KEhelp}
\textit{\textup{[}Lemma 4.1 in \cite{CE11}\textup{]}}
For all $0< m \leq M < \infty $ and all $t \geq 2$, there exists a positive constant $C$ such that for all $m \leq a \leq M$ and all $b \geq -a$, 
\begin{equation*}
        (t-1)a^{t-2}b^2 \leq (a+b)^t -a^t -ta^{t-1}b \leq Cb^2(1+|b|^{t-2}). 
    \end{equation*} 
Note that $|(a+b)^t -a^t -ta^{t-1}b| \leq C_1 b^2+ C_2|b|^t$, where $C_1, C_2$ only depends on $C,M,t$.
\end{lemma}
% C = 2^{t-3}t(t-1)max{1, M^{t-2}}

\begin{lemma}
\label{lem:DCThelp}
Let $u_{2,L}, u_{2},$ and $u_1$ be the unique positive ground state solution for the square root of electronic density corresponding to nuclear densities $m_{2,L}, m_{2},$ and $m_1$, respectively. Then $$\lim_{L \rightarrow \infty}\int_{\Gamma_L}  u_{2,L}^{10/3} - u_1^{10/3} = \int_{\Gamma_\infty}u_2^{10/3} - u_1^{10/3}.$$ 
\end{lemma}
\begin{proof}
We begin by noting that
\begin{equation} \label{eq:DCThelp_0}
\begin{split}
\left | \int_{\Gamma_\infty} \left(\chi_{\Gamma_L} u_{2,L}^{10/3}  - \chi_{\Gamma_L} u_1^{10/3} - u_2^{10/3} + u_1^{10/3} \right) \right | &\leq        \left|\int_{\Gamma_\infty} \chi_{\Gamma_L} u_{2,L}^{10/3} - \chi_{\Gamma_L} u_2^{10/3} \right|  \\ &+  \left|\int_{\Gamma_\infty} \chi_{\Gamma_L} u_2^{10/3} - \chi_{\Gamma_L} u_1^{10/3} - u_2^{10/3} + u_1^{10/3}) \right | \\
&\leq \left |\int_{\Gamma_L}u_{2,L}^{10/3} - u_2^{10/3}\right | + \left |\int_{\Gamma_\infty \setminus \Gamma_L} u_2^{10/3} - u_1^{10/3}\right |.
\end{split}
\end{equation}
Thus, it suffices to show that both the terms on the right hand side vanish in the limit $L \to \infty$ to establish the lemma.

To establish the convergence of the first term in the inequality~\eqref{eq:DCThelp_0}, we observe that using equation~\eqref{eq:dens_decay} it holds that for all $y \in \Gamma_L$,
 $$  | (u_{2,L} - u_2)(y)| \leq c_1e^{-c_2 dist(y, \partial \Omega_{2L-L_0})},
$$
where $c_1, c_2$ are positive constants. Since $\inf u_2 > 0$ and belongs to $L^\infty(\mathbb{R}^3)$, for large enough $L$, 
$$ (-c_1e^{-c_2 dist(y, \partial \Omega_{2L-L_0})} + u_2(y))^{10/3} \leq (u_{2,L}(y))^{10/3} \leq (u_2(y) + c_1e^{-c_2 dist(y, \partial \Omega_{2L-L_0})})^{10/3}.$$
Using Lemma~\ref{lem:KEhelp}, this implies that, for any $y \in \Gamma_L$,
$$
 -\frac{10}{3}f(y)\|u_2\|_{L^\infty(\mathbb{R}^3)}^{7/3} + \frac{7}{3}f(y)^2 ( \inf u_2)^{4/3} \leq u_{2,L}^{10/3} - u_2^{10/3}  \leq \frac{10}{3}f(y)\|u_2\|_{L^\infty(\mathbb{R}^3)}^{7/3}  + Cf(y)^2 + Cf(y)^{10/3},
$$
where $f(y) = c_1e^{-c_2 dist(y, \partial \Omega_{2L-L_0})}$ and $C$ is a positive constant. Note that for any positive $k$, $\lim_{L \rightarrow \infty} \int_{\Gamma_L} f^k(y) dy = 0$. Integrating the two inequalities shown above over $\Gamma_L$ and taking the limit $L\to\infty$, we obtain
\begin{equation}
    \label{eq:DCThelp_t1}
    \lim_{L \rightarrow \infty}\int_{\Gamma_L}  u_{2,L}^{10/3} - u_2^{10/3} = 0.
\end{equation}
To prove convergence of the second term in the inequality~\eqref{eq:DCThelp_0}, we proceed similarly using Lemma~\ref{thm:density_decay} and see that for all $x \in \Gamma_\infty \setminus \Gamma_{L_0}$,
$$| (u_2-u_1)(x)| \leq c_3e^{-c_4 dist(x, \Gamma_{L_0})},
$$
where $c_3, c_4$ are positive constants. We therefore get, for large enough $L$, 
$$ (-c_3e^{-c_4 dist(x, \Gamma_{L_0})} + u_2(x))^{10/3} \leq (u_{1}(x))^{10/3} \leq (u_2(x) + c_3e^{-c_4 dist(x, \Gamma_{L_0})} )^{10/3}.$$
Using Lemma~\ref{lem:KEhelp} again, we obtain, for every $x \in \Gamma_\infty \setminus \Gamma_{L_0}$, 
$$
 -\frac{10}{3}g(x)\|u_2\|_{L^\infty(\mathbb{R}^3)}^{7/3} + \frac{7}{3}g(x)^2 ( \inf u_2)^{4/3} \leq u_{1}^{10/3} - u_2^{10/3}  \leq \frac{10}{3}g(x)\|u_2\|_{L^\infty(\mathbb{R}^3)}^{7/3}  + Dg(x)^2 + Dg(x)^{10/3},
$$
where $g(x) = c_3e^{-c_4 dist(x, \Gamma_{L_0})}$ and $D$ is some positive constant. Note that for any positive $k$, $\lim_{L \rightarrow \infty} \int_{\Gamma_\infty \setminus \Gamma_L} g^k(x)dx = 0$. Integrating these inequalities over $\Gamma_\infty \setminus \Gamma_L$ and taking the limit $L \to\infty$, we obtain 
\begin{equation}
    \label{eq:DCThelp_t2}
        \lim_{L \rightarrow \infty}\int_{\Gamma_\infty \setminus \Gamma_L}  u_2^{10/3} - u_1^{10/3} = 0.
\end{equation}
Using equations~\eqref{eq:DCThelp_t1} and~\eqref{eq:DCThelp_t2} in the inequality~\eqref{eq:DCThelp_0}, we conclude that $\int_{\Gamma_L} ( u_{2,L}^{10/3} - u_1^{10/3} )$ converges to $ \int_{\Gamma_\infty} (u_2^{10/3} - u_1^{10/3})$ in the limit $L\to\infty$. 
\end{proof}

\begin{lemma}
\label{lem:lapgaminf}
    Let $w \in H^1_{per}(\Gamma_\infty)$. Then $w=0$ almost everywhere is a unique solution for $\nabla^2 w = 0$ in the sense of distributions. 
\end{lemma}
\begin{proof}
    Using Weyl's lemma, we conclude that $w$ is smooth and $\nabla^2 w = 0$ pointwise almost everywhere in $\mathbb{R}^3$. 
    Define $w_\epsilon$, for some $\epsilon > 0$, as $w_\epsilon(x) = w(x)f_\epsilon(x_3)$ as in \emph{Step-1} of the proof of theorem~\ref{thm:reflexivemin}. Since 
    $ \nabla \cdot (w_\epsilon \nabla w) = w_\epsilon \nabla^2 w + \nabla w_\epsilon \cdot \nabla w = \nabla w_\epsilon \cdot \nabla w$, we get 
    $$\int_{\Gamma_\infty} \nabla w_\epsilon \cdot \nabla w = \int_{\Gamma_\infty} \nabla \cdot (w_\epsilon \nabla w) = 0. $$
    In showing that the integral on the right is zero, we have used the divergence theorem along with the fact that $w_\epsilon$ and $w$ are periodic in $x_1, x_2$, and $w_\epsilon$ is compact in $x_3$. Noting that
    $$ \int_{\Gamma_\infty} \nabla w_\epsilon \cdot \nabla w = \int_{\Gamma_\infty} \nabla (w_\epsilon - w) \cdot \nabla w + \int_{\Gamma_\infty} \|\nabla w\|^2 = 0,$$
    and using the fact that $w \in H^1(\Gamma_\infty)$ and $w_\epsilon$ converges to $w$ in $H^1(\Gamma_\infty)$, we see that
    $$\int_{\Gamma_\infty} \|\nabla w\|^2 = 0.$$
    This implies $w$ is a constant almost everywhere on $\Gamma_\infty$. Since $w \in L^2(\Gamma_\infty)$ we get $w=0$ almost everywhere on $\mathbb{R}^3$, thereby establishing uniquenesss of the solution $w = 0$.  

\end{proof}

\begin{lemma}
\label{lem:Lposdef}
Let $m_{1, H}$ be the nuclear density defined in section~\ref{sec: minproblem}. Define $f_H = \frac{5}{3}\Tilde{u}_{1,H}^{4/3} + \psi_{1,H}$, where $(\Tilde{u}_{1,H}, \psi_{1, H})$ is the ground state solution corresponding to $m_{1, H}$.  Then:
\begin{enumerate}[(i)]
    \item For any $g \in H^1_{per}(\Gamma_\infty)$, $$\int_{\Gamma_\infty} \|\nabla g\|^2 + \int_{\Gamma_\infty}f_H g^2 \geq 0.$$
    \item For any $g_1 \in H^1_{per}(\Gamma_\infty) $ and $g_2 \in H^2_{per}(\Gamma_\infty)$,
    \begin{equation*}
    \int_{\Gamma_\infty} \lVert\nabla g_1\rVert^2 + \int_{\Gamma_\infty}f_Hg_1^2 \geq \int_{\Gamma_\infty} \lVert\nabla g_2\rVert^2 + \int_{\Gamma_\infty}f_Hg_2^2 - 2 \int_{\Gamma_\infty} (g_1-g_2)\nabla^2 g_2 + 2\int_{\Gamma_\infty}f_H(g_1-g_2)g_2. 
\end{equation*}
\end{enumerate}

\end{lemma}
\begin{proof}
Define $\Gamma^K_\infty = (-\frac{K}{2}, \frac{K}{2}) \times (-\frac{K}{2}, \frac{K}{2}) \times (-\infty, \infty)$ for any odd number, $K$ . It follows from the smooth version of Urysohn's lemma that there exists a smooth function $\xi_K$ such that
$$\xi_K (x)=\begin{cases}
1, & x \in \Gamma^K_\infty,\\ 
0, & x \in \mathbb{R}^3 \setminus \Gamma^{K+2}_\infty , \\
h(x), & x \in \Gamma^{K+2}_\infty \setminus \Gamma^{K}_\infty , \\
\end{cases} $$
where $h$ is a non-negative function defined on $\Gamma_K$ such that it is bounded by $1$ and its gradient is bounded by some positive constant $M$, independent of $K$. It is easy to see that $g\xi_K$ is weakly differentiable. Using the periodicity of $g\xi_K$, we see that for any $f \in L^\infty(\mathbb{R}^3)$, 
\begin{equation}
    \label{eq:Lpos_loc1}
    \int_{\Gamma^{K+2}_\infty \setminus \Gamma^{K}_\infty} fg^2\xi_K^2 \leq 4(K+1)\|f\|_{L^\infty(\mathbb{R}^3)}\int_{\Gamma_\infty} g^2.
\end{equation}
A similar argument using the function $f\equiv1$ informs us that $g\xi_K \in L^2(\mathbb{R}^3).$
We further have that
\begin{equation}
\label{eq:Lpos_loc2}
\begin{split}
      \int_{\Gamma^{K+2}_\infty \setminus \Gamma^{K}_\infty} \|\nabla (g \xi_K)\|^2 &\leq 2 \int_{\Gamma^{K+2}_\infty \setminus \Gamma^{K}_\infty} \|\nabla \xi_K\|^2 g^2+ 2\int_{\Gamma^{K+2}_\infty \setminus \Gamma^{K}_\infty} \|\nabla g\|^2 \xi_K^2 \\
    &\leq 8(K+1)(M^2\int _{\Gamma_\infty}g^2 + \int _{\Gamma_\infty} \|\nabla g\|^2).  
\end{split}  
\end{equation}
This implies that $g\xi_K \in H^1(\mathbb{R}^3).$

From section~\ref{subsec:thinfilmmod} we note that
$-\nabla^2 \Tilde{u}_{1,H} + f_H\Tilde{u}_{1,H} =0$ in the sense of distribution, $\Tilde{u}_{1,H} > 0$ and $\Tilde{u}_{1,H}$ belongs to $H^1_{loc}(\mathbb{R}^3)$. Moreover, $f_H \in H^1_{loc} \cap L^\infty (\mathbb{R}^3)$ and is periodic on $\Gamma_\infty$. Hence, using Lemma 6.2 in ~\cite{NO17}  we get that $\int_{\mathbb{R}^3} \|\nabla (g \xi_K)\|^2 + \int_{\mathbb{R}^3}f_H g^2\xi_K^2 \geq 0$ for every odd $K$. Suppose now that $\int_{\Gamma_\infty} \|\nabla g\|^2 + \int_{\Gamma_\infty}f_H g^2 = -\beta$, for some $\beta >0$. This implies that
$$- K^2\beta + \int_{\Gamma^{K+2}_\infty \setminus \Gamma^{K}_\infty} |\nabla g \xi_K|^2 + \int_{\Gamma^{K+2}_\infty \setminus \Gamma^{K}_\infty}f g^2\xi_K^2 \geq 0$$ for every odd $K$. Using equation~\eqref{eq:Lpos_loc1} and \eqref{eq:Lpos_loc2} we get
$$ -K^2\beta + (4K+1)\|f_H\|_{L^\infty(\mathbb{R}^3)}\int _{\Gamma_\infty}g^2 + 8(K+1)M^2 \int _{\Gamma_\infty}g^2 + 8(K+1)\int _{\Gamma_\infty} \|\nabla g\|^2 \geq 0.
$$
But since this inequality won't be satisfied by all $K$, there is no $\beta > 0$ such that $\int_{\Gamma_\infty} \|\nabla g\|^2 + \int_{\Gamma_\infty}f_H g^2 = -\beta$. Hence,
$$\int_{\Gamma_\infty} \|\nabla g\|^2 + \int_{\Gamma_\infty}f_{H} g^2 \geq 0.$$
To show the second part of the lemma, apply the foregoing inequality to $g = g_1 - g_2$ to get
\begin{displaymath}
\begin{split}
    \int_{\Gamma_\infty} \lVert\nabla g_1\rVert^2 + \int_{\Gamma_\infty}f_Hg_1^2 \geq& \int_{\Gamma_\infty} \lVert\nabla g_2\rVert^2 + \int_{\Gamma_\infty}f_Hg_2^2 \\&+ 2 \int_{\Gamma_\infty} \nabla(g_1-g_2)\cdot \nabla g_2 + 2\int_{\Gamma_\infty}f_H(g_1-g_2)g_2. 
\end{split}
\end{displaymath}
Since $g_1-g_2 \in H^1_{per}(\Gamma_\infty)$ and $g_2 \in H^2_{per}(\Gamma_\infty)$, it is easy to see that
$\int_{\Gamma_\infty} \nabla(g_1-g_2)\cdot \nabla g_2 = - \int_{\Gamma_\infty} (g_1-g_2) \nabla^2 g_2$, thereby proving the second part of the lemma. 
\end{proof}

\begin{lemma}
\label{lem:denconvg_help}
    Consider a sequence of functions $\{g_L\}_{L=1}^{\infty}$ in $H^1_{per}(\Gamma_L)$ such that $\|g_L\|_{H^1(\Gamma_L)} < C$ along a subsequence, where $C$ does not depend on $L$. Then $g_L$ converges, along a subsequence, to $g \in H^1_{loc}(\mathbb{R}^3)$ weakly in $H^1_{loc}(\mathbb{R}^3)$, strongly in $L^p_{loc}(\mathbb{R}^3)$ for $1 \leq p < 6$, and almost everywhere on $\mathbb{R}^3$.     
\end{lemma}
\begin{proof}
    Let $\Gamma_L^a = (-\frac{a}{2},\frac{a}{2}) \times (-\frac{a}{2},\frac{a}{2})  \times (-\frac{L}{2},\frac{L}{2})$ for any odd numbers $a$ and $L$. 
    Then $\|g_L\|_{H^1(\Gamma^a_L)} \leq aC$ for all $L \in \mathcal{I}$, where $\mathcal{I}$ is the subsequence for which $\|g_L\|_{H^1(\Gamma_L)} < C$.  
    Consider a monotonically increasing sequence of real numbers $\{R_k\}_{k=1}^\infty$ such that $R_L < L$ for all $L \in \mathcal{I}$. We then have, for $k\leq L$, 
    $$ \|g_{L}\|_{H^1(\Gamma^a_{R_k})} \leq \|g_{L}\|_{H^1(\Gamma^a_{R_L})} \leq \|g_{L}\|_{H^1(\Gamma^a_{L})} \leq aC.$$
    Fixing a value for $k$ and varying $L \ge k$, we see from the above inequality that $g_{L}|_{\Gamma^a_{R_k}}$ converges along a subsequence to some $\Tilde{g}_k \in H^1(\Gamma^a_{R_k})$ weakly in $H^1(\Gamma^a_{R_k})$, strongly in $L^p(\Gamma^a_{R_k})$ for $1 \leq p < 6$, and almost everywhere on $(\Gamma^a_{R_k})$. Moreover, since weak convergence in $H^1(\Gamma^a_{R_k})$ implies $\|\Tilde{g}_k\|_{H^1(\Gamma^a_{R_k})} \leq \liminf \|g_{L}\|_{H^1(\Gamma^a_{R_k})}$, we get
    $$\|\Tilde{g}_k\|_{H^1(\Gamma^a_{R_k})} \leq aC.$$
Let us denote this subsequence as $\{g_{L^k_{n}}\}_{n=1}^\infty$. Using an argument similar to that just presented, we can extract a subsequence $\{g_{L^{k+1}_{n}}\}_{n=1}^\infty$ of $\{g_{L^k_{n}}\}_{n=1}^\infty$ such that its restriction to $\Gamma^a_{R_{k+1}}$ converges to $\tilde{g}_{k+1} \in  H^1(\Gamma^a_{R_{k+1}})$ weakly in $H^1(\Gamma^a_{R_{k+1}})$, strongly in $L^p(\Gamma^a_{R_{k+1}})$ for $1 \leq p < 6$, almost everywhere on $\Gamma^a_{R_{k+1}}$, and such that
    $$ \|\Tilde{g}_{k+1}\|_{H^1(\Gamma^a_{R_{k+1}})} \leq aC. $$
    Note that $\Tilde{g}_{k+1}|_{\Gamma^a_{R_{k}}} = \Tilde{g}_{k}$. Following this reasoning, we define a function $g:\mathbb{R}^3 \to \mathbb{R}$ such that $g|_{\Gamma^a_{R_{k}}} = \Tilde{g}_{k}$. We then get, for any $k \in \mathbb{N}$, that 
    $$ \| g \|_{H^1(\Gamma^a_{R_{k}})} \leq aC.$$
    Taking the limit $k \to \infty$ we get,
    $$ \|g\|_{H^1(\Gamma^a_\infty)} \leq aC,$$
    and thus establishes that $g \in H^1_{loc}(\mathbb{R}^3).$ To prove the lemma, we use a diagonal argument by noting that for any $m \in \mathbb{N}$, $\{g_{L^k_{k}}\}_{k=1}^\infty$ converges to $\Tilde{g}_{m}$ weakly in $H^1(\Gamma^a_{R_m})$, strongly in $L^p(\Gamma^a_{R_m})$ for $1 \leq p < 6$ and almost everywhere on $\Gamma^a_{R_m}$. This establishes the existence of a subsequence of $\{g_L\}_{L=1}^{\infty}$ that converges to $g$ weakly in $H^1_{loc}(\mathbb{R}^3)$, strongly in $L^p_{loc}(\mathbb{R}^3)$ for $1 \leq p < 6$ and almost everywhere on $\mathbb{R}^3$.
\end{proof}

\section*{Acknowledgments}
The authors are grateful to Vikram Gavini, as well as to the anonymous referees, for useful comments and suggestions.

\appendix

\section{Thermodynamic Limit for TFW models} \label{app:tfw_thms}
We summarize here a few theorems regarding the thermodynamic limit in the context of TFW models that are used in this work. 

\subsection{Thermodynamic Limit for Perfect Crystals}
We recall in this setting the following theorem which establishes the existence and uniqueness of solutions to the coupled system of PDEs \eqref{eq:TFW_smeared_EL}:

\begin{theorem} \label{thm:TFW_smeared_EL_exst_uniq}
\textit{\textup{[}Theorem 6.5 in \cite{CLBL98}\textup{]}}
Given $n \in \mathcal{M}$, there exists a unique solution $(u,\phi) \in L^{7/3}_{\text{loc}}\cap L^2_{\text{unif}}(\mathbb{R}^3) \times L^1_{\text{unif}}(\mathbb{R}^3)$ to the Schr\"odinger-Poisson system \eqref{eq:TFW_smeared_EL} such that $u \geq 0$. Further, $\inf u > 0$ and $u \in L^\infty \cap H^2_{\text{loc}} \cap C^{0,\alpha}(\mathbb{R}^3)$ for all $0 < \alpha < 1$. 
\end{theorem}
Note that if the nuclear distribution $m$ is periodic, then the uniqueness of the solution of the coupled system \eqref{eq:TFW_smeared_EL} implies that both $\phi$ and $u$ are periodic with the same periodicity as $m$. 

\begin{corollary} \label{cor:ortner_TFW_periodic}
\textit{\textup{[}Theorem 3.1 in \cite{NO17}\textup{]}}
Let $n \in \mathcal{M}^M$. Then the ground state solution $(u, \phi)$ given by theorem~\ref{thm:TFW_smeared_EL_exst_uniq} are uniformly bounded in $H^4_{\text{unif}}(\mathbb{R}^3) \times H^2_{\text{unif}}(\mathbb{R}^3)$.    
\end{corollary}

The finite nuclei system can also be generalised in a similar fashion. 
%proposition 6.1
\begin{proposition}\label{prop:ortner_TFW_general_finite}
\textit{\textup{[}Proposition 6.1 in \cite{NO17}\textup{]}}
Let $m : \mathbb{R}^3 \rightarrow \mathbb{R}_{\geq 0} $ satisfy $\|m\|_{L^2_{unif}(\mathbb{R}^3)} \leq M$. Then define $m_{R_n} = m\chi_{B_{R_n}(0)}$, where $R_n \uparrow \infty$. Then the unique solution to the minimisation problem 
\begin{displaymath}
    \begin{split}
          I(m_{R_n}) &= \inf\; \left\{ E (v, m_{R_n}) \;:\; v \in H^1(\mathbb{R}^3), v\geq 0, \int_{\mathbb{R}^3}v^2 = \int_{\mathbb{R}^3}m_{R_n}\right\},\\
    E (v, m_{R_n}) &=  \int_{\mathbb{R}^3} \lVert \nabla v(x) \rVert^2\,dx + \int_{\mathbb{R}^3} v^{\frac{10}{3}}(x)\,dx\\
    & + \frac{1}{2}\int_{\mathbb{R}^3}\int_{\mathbb{R}^3} \frac{(v^2(x) -m_{R_n}(x))(v^2(x) -m_{R_n}(x))}{\lVert x - y\rVert} \, dx \,dy ,\\  
    \end{split}
\end{displaymath}
satisfy
\begin{displaymath}
     \begin{split}
        &-\nabla^2 u_{R_n} + \frac{5}{3}u_{R_n}^{5/3} + \phi_{R_n}u_{R_n} = 0, \\
    &-\nabla^2 \phi_{R_n} = 4\pi(u_{R_n}^2 - m_{R_n}).  
    \end{split}   
\end{displaymath}

Moreover, $(u_{R_n}, \phi_{R_n})$ is uniformly bounded in $H^4_{\text{unif}}(\mathbb{R}^3) \times H^2_{\text{unif}}(\mathbb{R}^3)$.
\end{proposition}

The thermodynamic limit arguments presented above can be readily generalized to the case of a \emph{supercell}, which is a volume that is larger than the unit cell, and which can be periodically tiled to fill $\mathbb{R}^3$. Specifically, the unit cell $\Gamma_0$ is replaced with a supercell $\Gamma$ with dimensions $L = (L_1, L_2, L_3)$. Further, define $\mathcal{R}_\Gamma = L_1\mathbb{Z} \times L_2\mathbb{Z} \times L_3\mathbb{Z}$ and its reciprocal lattice $\mathcal{R}^*_\Gamma = 2\pi \Big(\frac{\mathbb{Z}}{L_1} \times \frac{\mathbb{Z}}{L_2} \times \frac{\mathbb{Z}}{L_3} \Big)$. We denote by $\mathcal{M}_L$ the set of nuclear densities in $\mathcal{M}$ that are $\mathcal{R}_L$-periodic. Further details can be found in \cite{CE11}. It is to be noted that this modification is commonly employed in numerical computations.

Given a periodic nuclear charge distribution  $m \in \mathcal{M}_L$, the TFW energy functional on the supercell is written as follows: for any $v \in H^1_{\text{per}}(\Gamma)$,
\begin{equation}
 \label{eq:TFW_energy_supercell}
 \begin{split}
  E^{L}_{m}(v) &= \int_{\Gamma} \lVert\nabla v\rVert^2 + \int_{\Gamma} v^{\frac{10}{3}}  + \frac{1}{2}D_{\Gamma}(m - v^2, m - v^2), \\
D_{\Gamma}(f, g) &= \iint_{{\Gamma} \times {\Gamma}} f(x)g(y)G_\Gamma(x-y)\,dx dy, \quad f, g \in L^{2}({\Gamma}).   
 \end{split}
\end{equation}

In the equation displayed above, $$G_\Gamma(x) = \frac{4\pi}{|\Gamma|}\sum_{k \in \mathcal{R}^*_{\Gamma} \setminus \{0\}} \frac{e^{ik\cdot x}}{\|k\|^2} + C,$$
where the constant $C$ is chosen such that the minimum of $G$ is zero.

It is easy to see that $G_\Gamma \in L^2_{per}(\Gamma)$ is the unique solution of the following PDE:
\begin{equation}
\label{eq:green_TFW_supercell}
    \begin{split}
      -\nabla^2 G_\Gamma &= 4\pi\left(-|\Gamma|^{-1} + \sum_{k \in \mathcal{R}_\Gamma } \delta_k \right), \\ \min_{\mathbb{R}^3}G_\Gamma = 0, & \quad G_\Gamma \hspace{2pt}
 \text{periodic over}\hspace{2pt} \Gamma.  
    \end{split}
\end{equation}

The ground state energy per neutral supercell is obtained by solving the following minimization problem:
\begin{equation} \label{eq:TFW_min_supercell}
I^{L}_m = \inf\left\{E^L_m(v)\;:\; v \in H^1_{per}(\Gamma), \int_{\Gamma}v^2 = \int_{\Gamma} m \right\} .
\end{equation}
We have the following result about the existence and uniqueness of solution to this minimization problem, in analogy with Theorem~\ref{thm:TFW_smeared_EL_exst_uniq}.
\begin{theorem} \label{thm:TFW_supercell_exst_uniq}
\textit{\textup{[}Proposition 2.1 in \cite{CE11}\textup{]}}
 Let $m \in L^2_{per}(\Gamma)$. Then there exists exactly two minimizers $u_\Gamma > 0$ and $-u_\Gamma$ for the minimization problem \eqref{eq:TFW_min_supercell}. Further, $u_\Gamma \in H^3_{per}(\Gamma) \hookrightarrow C^1(\mathbb{R}^3) \cap L^{\infty}(\mathbb{R}^3)$ and satisfies the Euler equation 
\begin{equation} \label{eq:TFW_supercell_EL}
-\nabla^2 u_\Gamma + \frac{5}{3}u_\Gamma^{7/3} + \phi_\Gamma u_\Gamma = 0,
\end{equation}
where $$\phi_\Gamma(x) = \int_{\Gamma} G_\Gamma(x-y) (u_\Gamma^2(y) - m_\Gamma(y))\,dy - \epsilon_\Gamma,$$ and $\epsilon_\Gamma$ is the Lagrange multiplier of the constraint $\int_{\Gamma} u^2 =\int_{\Gamma} m$. The potential $\phi_\Gamma$ is the unique solution in $H^1_{per}(\Gamma)$ to the periodic Poisson problem:
\begin{equation} \label{eq:TFW_supercell_potential}
\begin{split}
-\nabla^2 \phi_\Gamma &= 4\pi (u_\Gamma^2 - m),\\
\int_{\Gamma}\phi_\Gamma &= 0. 
\end{split}
\end{equation}
In addition, there exist constants  $0 < m \leq M < \infty$ such that $m \leq u_L \leq M$.
\end{theorem}

\subsection{TFW theory for polymers and thin Films}
\label{sec:thinfilm}
We collect the relevant results related to thin films that will be needed later in the study of  quasi-planar defects.

\subsubsection{Polymers}
\label{subsec:polymers}
Consider a positive bounded nuclear distribution $m_{a,b}$ contained within $\mathbb{R} \times [-a/2, a/2) \times [-b/2, b/2)$, for some $a, b > 0$, such that $m_{a,b}$ is periodic in $x_1$ with periodicity $1$. We denote by $\Gamma^P_\infty= [-\frac{1}{2},\frac{1}{2}) \times \mathbb{R}  \times \mathbb{R}$ the corresponding unit cell. We denote by $L^p_{per}(\Gamma^P_\infty)$ the space of $L^p$ functions periodic along $x_1$, with a similar definition for Sobolev spaces. 

The energetics of the polymer model is posed as the following minimization problem in \cite{BTF00}:
 \begin{equation}
     \label{eq:polymermin}
     I^P_H = \inf_{w \in X^{m^P_H}} E^{m^P_H}_H(w),
 \end{equation}

 where, 
 \begin{equation} \label{eq:polymer_green}
 \begin{split}
 E^{m^P_H}_H(u) &= \int_{\Gamma_\infty} \|\nabla u\|^2 + \int_{\Gamma_\infty} |u|^{10/3} + \frac{1}{2} \int_{\Gamma_\infty} (G^P *_{\Gamma_\infty}(u^2-m^P_H)) (u^2-m^P_H), \\
 X^{m^P_H} &= \left\{w \in H^1_{per}(\Gamma^P_\infty) \;:\; (\log(2+\|x\|))^{1/2} w \in L^2(\Gamma^P_\infty), \int_{\Gamma^P_\infty} w^2 = \int_{\Gamma^P_\infty} m_H\right\},\\
 G^P(x) &= -2\log\Big(\sqrt{x_1^2+x_2^2}\Big) + \sum_{k \in \mathbb{Z} \times \{0\} \times \{0\}}\left( 
 \frac{1}{\lVert x-k\rVert} -\int_{(-\frac{1}{2}, \frac{1}{2}) \times \{0\} \times \{0\}}\frac{dy}{\lVert x-k-y\rVert}   \right).
 \end{split}
 \end{equation}

We have the following result:
\begin{theorem}
\label{thm:polymer}
\textit{\textup{[}Theorem 2.4 in \cite{BTF00}\textup{]}}
    Let $\int_{\Gamma^P_\infty}m_{a,b}(x)\,dx$ be finite. Then the system,
    \begin{displaymath}
        \begin{split}
           &-\nabla^2 u + \frac{5}{3}u^{7/3} + u\phi =0, \\
        &-\nabla^2 \phi = 4\pi (u^2-m_{a,b}),\\
        &u\geq0,  
        \end{split}   
    \end{displaymath}
    has a unique solution $(u_{a,b}, \phi_{a,b})$ in $(L^2_{\text{unif}} \cap L^{7/3}_{\text{loc}}(\mathbb{R}^3)) \times L^1(\mathbb{R}).$ which is periodic in $x_1$. Moreover $u_{a,b}$ is uniformly bounded almost everywhere and there exists a positive constant, $C$, independent of $a$ and $b$ such that $u_{a,b}(x) \leq \frac{C}{(x_1^2+x_2^2)^{3/4}}$ for almost every $x$ outside support of $m_{a,b}$.
\end{theorem}

Using equation $(2.53)$ in \cite{BTF00}, we can also deduce that $|\nabla u_{a,b}(x)| \leq C/(x_1^2+x_2^2)^{3/4}$ for almost every $x$ outside support of $m_{a,b}$. Note that $(u_{a,b}, \phi_{a,b})$ are derived from taking thermodynamic limit of solutions of finite TFW models along any Van Hove sequence. It also holds that $(u_{a,b}, \phi_{a,b})$ is uniformly bounded in $H^4_{\text{unif}}(\mathbb{R}^3) \times H^2_{\text{unif}}(\mathbb{R}^3)$. In addition, if we assume $m_{a,b}$ to be smooth then $\phi_{a,b}$ is also uniformly bounded almost everywhere.

\subsubsection{Thin films}
\label{subsec:thinfilm}

Consider a positive bounded nuclear distribution $m_H$ contained within $\mathbb{R} \times \mathbb{R} \times [-H/2, H/2)$, for some $H > 0$, such that $m_H$ is periodic in $x_1$ and $x_2$ with periodicity $1$, and is symmetric with respect to $x_1$ and $x_2$. We denote by $\Gamma_\infty= [-\frac{1}{2},\frac{1}{2}) \times [-\frac{1}{2},\frac{1}{2})  \times \mathbb{R}$ the corresponding unit cell. We denote by $L^p_{per}(\Gamma_\infty)$ the space of $L^p$ functions periodic along $x$ and $y$, with a similar definition for Sobolev spaces. The energetics of the thin film model is posed as the following minimization problem in \cite{BTF00}:
\begin{equation}
    \label{eq:thinfilmmin}
    I_H = \inf_{w \in X^{m_H}} E^{m_H}_H(w),
\end{equation}
where, 
\begin{equation} \label{eq:thinfilm_green}
\begin{split}
E^{m_H}_H(u) &= \int_{\Gamma_\infty} \|\nabla u(x)\|^2\,dx + \int_{\Gamma_\infty} |u|^{10/3}(x)\,dx\\
 &+ \frac{1}{2} \int_{\Gamma_\infty}\int_{\Gamma_\infty} G(x - y) (u^2(x)-m_H(x)) (u^2(y)-m_H(y))\,dxdy, \\
X^{m_H} &= \left\{w \in H^1_{per}(\Gamma_\infty) \;:\; (1+\|x\|)^{1/2}w \in L^2(\Gamma_\infty), \int_{\Gamma_\infty} w^2 = \int_{\Gamma_\infty} m_H\right\},\\
G(x) &= -2\pi|x_3| + \sum_{k \in \mathbb{Z}^2 \times \{0\}}\left( 
\frac{1}{\lVert x-k\rVert} -\int_{(-\frac{1}{2}, \frac{1}{2})^2 \times \{0\}}\frac{dy}{\lVert x-k-y\rVert}   \right).
\end{split}
\end{equation}

The minimization problem \eqref{eq:thinfilmmin} is well defined since, $\int_{\Gamma_\infty}\int_{\Gamma_\infty} G(x-y)f(x)g(y)\,dxdy$ is well defined for all $f,g \in L^1(\Gamma_\infty) \cap L^2_{\text{unif}}(\mathbb{R}^3)$ such that $|x_3|f,|x_3|g \in L^1(\Gamma_\infty)$.  Note that $|G(x) + 2\pi |x_3| - \frac{1}{|x|}| \in L^\infty(\mathbb{R}^3)$ and $w \in H^1_{per}(\Gamma_\infty)$ implies $w \in L^p_{\text{unif}}(\mathbb{R}^3)$ for all $1 \leq p \leq 6$. 
% In proof of theorem 3.2% 

We then have the following result:

\begin{theorem}
\textit{\textup{[}Theorem 3.2 in \cite{BTF00}\textup{]}} \label{thm:thinfilm}
    There exists exactly two minimizers $u_H > 0$ and $-u_H$ for the minimization problem \eqref{eq:thinfilmmin}, and
    a real number $d_H$ such that,
    $ \phi_H = \int_{\Gamma_\infty} G(x -y)(u_H^2(y)-m_H(y))\,dy + d_H$, 
    where $G$ is as specified in \eqref{eq:thinfilm_green}. Further, $u_H(x)$ satisfies
    \begin{equation}
    \label{eq:thin_film_EL}
    -\nabla^2 u_H + \frac{5}{3}u_H^{7/3} + u_H\phi_H =0     
    \end{equation}
     in the sense of distributions, and there exists a positive constant $C$ not depending on the nuclear density such that 
     $$  u_H(x) \leq \frac{C}{|x_3|^{3/2}}, \hspace{4pt} \forall |x_3|\geq H .$$
     Moreover, $u_H$ is uniformly bounded almost everywhere and for any $R> 0$ there exists a positive constant $\nu_R$ such that $\inf_{|x_3|< R} u_H \geq \nu_R$.
\end{theorem}

$I_H$ gives the energy of a thin film up to a constant that depends only on the nuclear density. Note that $\phi_H(x)$  satisfies $$ -\nabla^2 \phi_H = 4\pi (u_H^2 - m_H) $$ in the sense of distributions. 

Now $(u_H, \phi_H)$ are derived from taking thermodynamic limit of solutions of finite TFW models along any Van Hove sequence. It also holds that $(u_{H}, \phi_{H})$ is uniformly bounded in $H^4_{\text{unif}}(\mathbb{R}^3) \times H^2_{\text{unif}}(\mathbb{R}^3)$. In addition, if we assume $m_H$ to be smooth
then $\phi_H$ is also uniformly bounded almost everywhere. 

The requirement of symmetry along both the $x_1$ and $x_2$ directions in Theorem~\ref{thm:thinfilm} is too stringent for application to planar defects such as stacking faults. We can, however, relax this requirement to impose symmetry only along one direction if we extend the polymer to thin film, as described earlier.

\section{Numerical algorithm for quasi-1D TFW models} \label{app:tfw_numerics}
For the numerical results presented in the text we follow an algorithm based on the augmented Lagrangian approach presented in \cite{SP14}. We present the details of the numerical algorithm here in a 1D context for completeness. The staggered algorithm consists of solving the constrained energy minimization problem for $u$ for a fixed $\phi$, and using a Poisson solver for $\phi$ at fixed $u$, and iterating this process until convergence. The charge neutrality constraint for the minimization problem is handled by introducing the following augmented Lagrangian: for real constants $\mu$, $c$, define 
\begin{equation} \label{eq:tfw_aug_Lagr}
\begin{split}
L(u, \mu, c) &= C_W \int_{\mathbb{R}} u'(z)^2\, dz + C_{TF} \int_{\mathbb{R}} u(z)^{10/3}\,dz + \int_{\mathbb{R}} (m(z) - u(z)^2)\phi(z)\,dz\\
 &+ \mu\left(\int_{\mathbb{R}} (u(z)^2 - m(z))\,dz\right) + \frac{1}{2c}\left(\int_{\mathbb{R}} (u(z)^2 - m(z))\,dz\right)^2,
\end{split}
\end{equation}
where the potential $\phi(z)$ satisfies
\begin{equation} \label{eq:tfw_potl}
-\phi''(z) = 4\pi(m(z) - u(z)^2), \quad z \in \mathbb{R}.
\end{equation}
The steps involved in the staggered algorithm listed in Algorithm~\ref{alg:tfw_alm}. Note that for all numerical calculations, we take $C_{TF}=3.2$ and $C_W = 1$.

Since the domain is infinitely large along the $z$ direction, a choice of appropriate boundary conditions is necessary for the practical implementation of the aforementioned minimization problem. This problem can be addressed using bulk boundary conditions, as we have demonstrated that the difference in the electronic fields between the perturbed and perfect crystals is bounded by an exponentially decaying function away from the defect boundary. Additionally, we show that electronic fields with periodic boundary conditions converge to the true electronic fields. Given the ease of implementing periodic boundary conditions, we apply them over a domain $[-L, L]$, where $L$ is chosen large enough to ensure that the perturbations decay within this region. 

We use a uniform $(x_i)_{i=0}^N$ over the domain $[-L,L]$ to locate points where the density and potential fields will be calculated; the algorithm works even for non-uniform grids too. The values of the square root electron density $u$ and the potential $\phi$ at these grid points are denoted as $(u_i)_{i=0}^N$ and $(\phi_i)_{i=0}^N$, respectively. The minimization problem in the first step of the loop in Algorithm~\ref{alg:tfw_alm} is computed using the gradient descent algorithm, and the final step requiring the solution of equation~\eqref{eq:tfw_potl} is performed using the finite difference method. 

\begin{algorithm}[h]
\caption{Staggered algorithm using the Augmented Lagrangian Method for solving the TFW model numerically}
\label{alg:tfw_alm}
\KwData{Mesh $(x_i)_{i=0}^N$, Nuclear density $(m_i)_{i=0}^N$, initial guess for square root electron density $u^0$, initial guess for potential $\phi^0$, initial guess for augmented Lagrangian parameters $\mu^0, c^0, \kappa$}
\KwResult{Square root electron density $(u_i)_{i=0}^N$, Potential $(\phi_i)_{i=0}^N$, updated augmented Lagrange parameters $\mu, c$}
$u \gets u^0,$
$\phi \gets \phi^0,$
$\mu \gets \mu^0,$
$c \gets c^0$\;
$k \gets 0$\;
\While{not converged}{
    $u^{k+1} = \text{argmin}_u\, L(u, \mu^k, c^k)$\;
    $\mu^{k+1} = \mu^k + \frac{1}{c^k}\left(\int (u(z)^2 - m(z))\,dz\right)$\;
    $c^{k+1} = \kappa c^k$\;
    Compute $\phi^{k+1}$ by solving equation \eqref{eq:tfw_potl} with $u \gets u^{k+1}$\;
    $k \gets k + 1$\;
}
\end{algorithm}
The solution of the minimization problem for $u$ requires the functional derivative of the augmented Lagrangian $L$ with respect to $u$. This is easily computed as
\begin{displaymath}
\partial_u L(u, \mu, c) = -2C_W u''(z) + \frac{10}{3}C_{TF} u(z)^{7/3} - 2u(z)\phi(z) + 2\mu u(z) + \frac{2}{c}\left(\int (u(z)^2 - m(z))\,dz\right)u(z).
\end{displaymath}
All the integrals are evaluated using the trapezoid rule.

Although Dirichlet and periodic boundary conditions both lead to the same electronic fields, one might question whether their rates of convergence differ. In \cite{VG15}, it is demonstrated that the energy converges more rapidly with periodic boundary conditions, while the convergence rate for the electronic fields remains the same. We observed the same trend in our numerical simulations. For this reason, we employ periodic boundary conditions for all the results shown in this work. 

As a simple test of the algorithm, we present numerical results for a simple jellium model. The reference nuclear distribution corresponds to a uniform charge density field, with density $1$. The perturbed nuclear density adds a uniform bump with density $11$ and width $0.05$ (arbitrary units) about the the origin, as shown in Figure~\ref{fig:jel_nuc_dens}. Note that the nuclear distribution varies only along the $x_3$ direction. We show in Figure~\ref{fig:dE_jellium} the effect of varying the width of simulation cell on the relative energy of the defect, clearly indicating the convergence of the staggered numerical algorithm.  
\begin{figure}[h]
\begin{subfigure}{0.49\textwidth}
\centering
\includegraphics[width=0.9\textwidth]{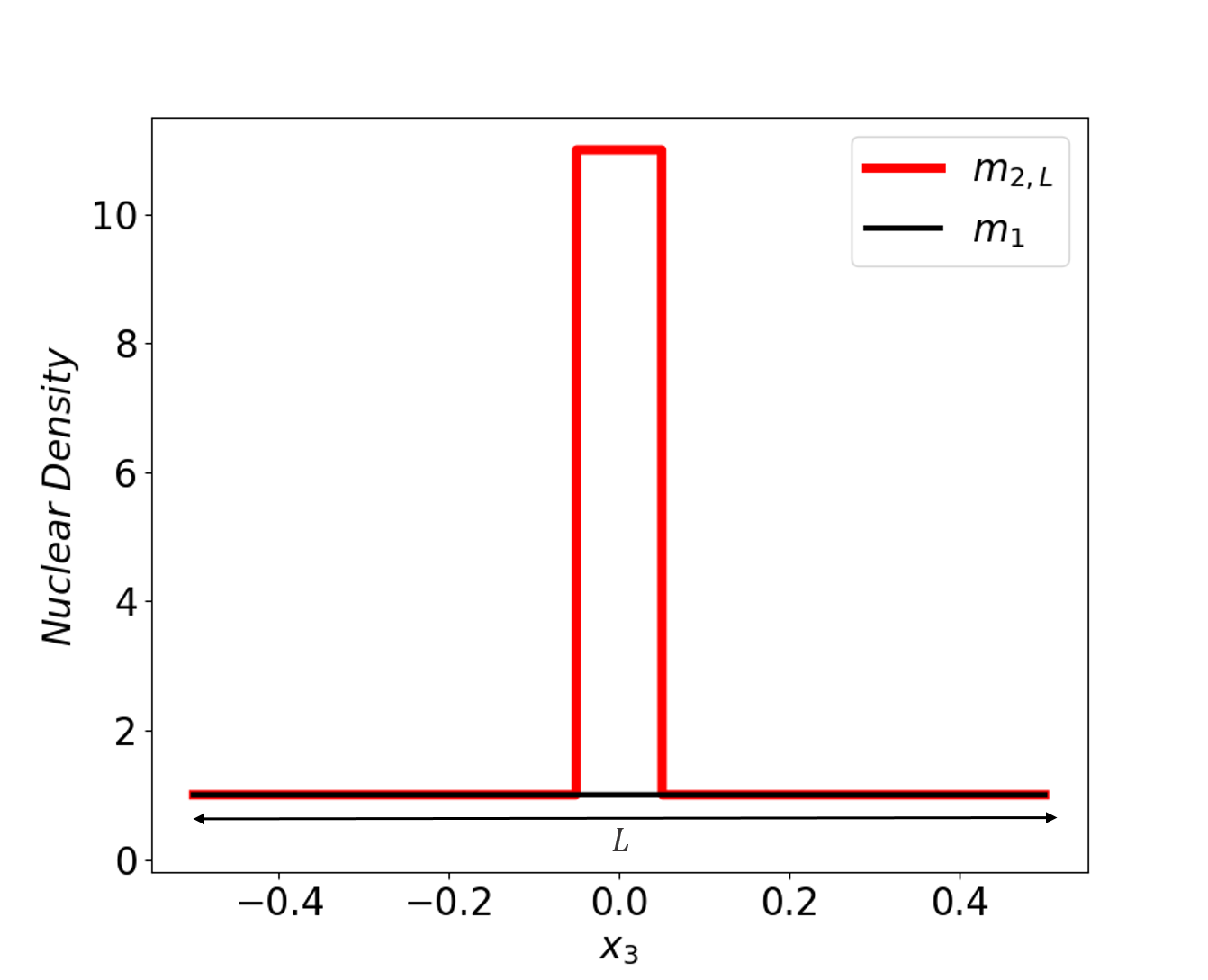}
\caption{Perfect and perturbed jellium configurations.}
\label{fig:jel_nuc_dens}
\end{subfigure}
~
\begin{subfigure}{0.49\textwidth}
\centering
\includegraphics[width=0.9\textwidth]{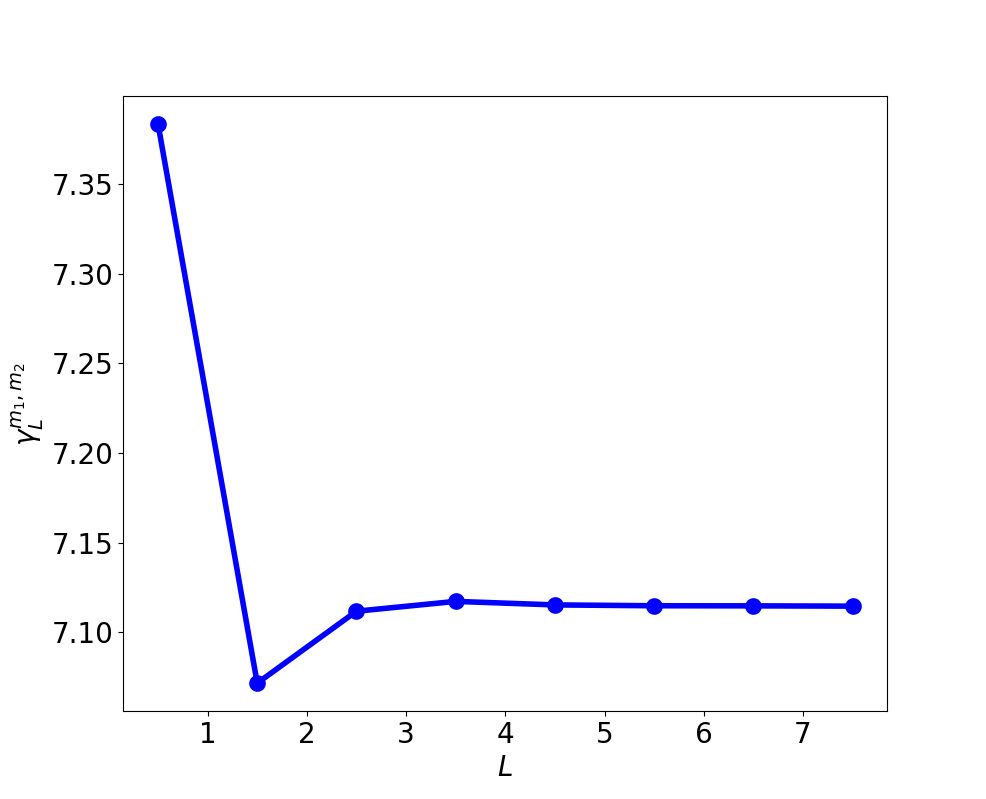}
\caption{Convergence of relative energy.}
\label{fig:dE_jellium}
\end{subfigure}
\caption{In figure~\ref{fig:jel_nuc_dens} we plot the nuclear density of the perfect and perturbed crystals over a finite domain; the density of the perfect crystal extends periodically, while that of the perturbed configuration is identical to the perfect crystal outside the central perturbed region. In figure~\ref{fig:dE_jellium} we show the convergence of the relative energy as a function of the domain size, to illustrate the convergence of the staggered algorithm.} 
\end{figure}
The decay of the difference in square root electron density, and the difference in potential, between the perturbed and reference jellium models is shown in Figure~\ref{fig:gsf_jellium}. We thus see that the perturbations decay exponentially as expected from the theoretical analysis, but the domain over which this decay occurs is typically much larger than the size of the domain where the reference nuclear density is perturbed.
\begin{figure}[h]
\begin{subfigure}{0.45\textwidth}
\centering
\includegraphics[width=0.9\textwidth]{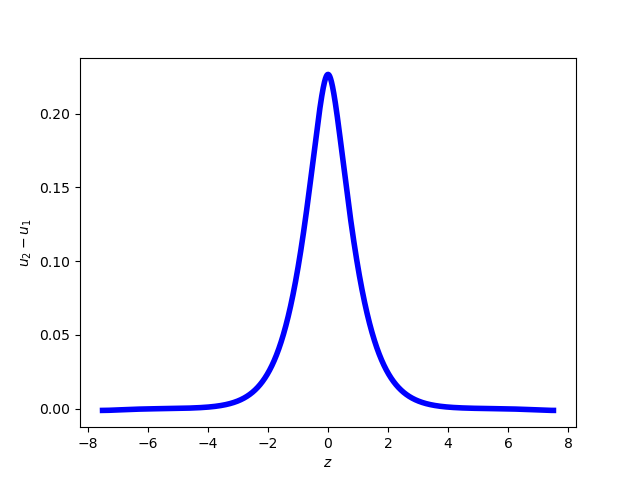}
\end{subfigure}
~
\begin{subfigure}{0.45\textwidth}
\centering
\includegraphics[width=0.9\textwidth]{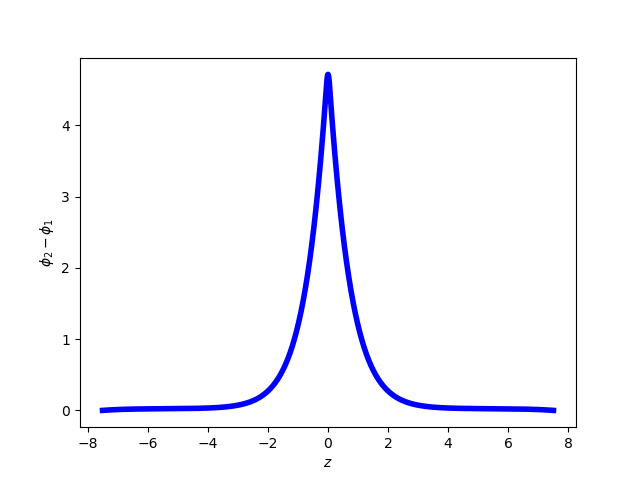}
\end{subfigure}
\caption{Decay of the difference in electronic fields between the perturbed and reference jellium models. We note for reference that the width of the perturbation to the reference nuclear density is $0.05$; both the perturbation to the square root electronic density and potential decay over much larger distances, as shown.}
\label{fig:gsf_jellium}
\end{figure}

\newpage
%\hspace{2cm}
%\bibliographystyle{amsalpha}
%\bibliography{tfw_refs}

\begin{thebibliography}{LKBK01}

\bibitem[BB00]{BTF00}
X.~Blanc and C.~Le Bris, \emph{{Thomas-Fermi type theories for polymers and
  thin films}}, Adv. Differential Equations \textbf{5} (2000), no.~7-9,
  977 -- 1032.

\bibitem[BBL81]{BBL81}
R.~Benguria, H.~Brezis, and E.~Lieb, \emph{The Thomas--Fermi--von Weizs\"acker
  theory of atoms and molecules}, Commun. Math. Phys. \textbf{79} (1981),
  167--180.

\bibitem[Bla06]{Blanc06}
X.~Blanc, \emph{Unique solvability of a system of nonlinear elliptic pdes
  arising in solid state physics}, SIAM J. Math. Anal. \textbf{38} (2006), no.~4, 1235--1248.

\bibitem[CE11]{CE11}
E.~Canc\`es and V.~Ehrlacher, \emph{Local defects are always neutral in the
  {T}homas--{F}ermi--von {W}eizs\"acker theory of crystals}, Arch. Rational Mech. Anal. \textbf{202} (2011), 933--973.

\bibitem[CLBL98]{CLBL98}
I.~Catto, C.~Le~Bris, and P-L. Lions, \emph{The mathematical theory of
  thermodynamic limits}, Oxford Mathematical Monographs, Oxford University
  Press, 1998.

\bibitem[DIG15]{VG15}
S. Das, M. Iyer, and V. Gavini, \emph{Real-space formulation of
  orbital-free density functional theory using finite-element discretization:
  The case for al, mg, and al-mg intermetallics}, Phys. Rev. B \textbf{92}
  (2015), 014104.

\bibitem[Dir30]{Dirac30}
P.A.M. Dirac, \emph{Note on exchange phenomena in the {T}homas atom},
  Math. Proc. Camb. Phil. Soc. \textbf{26}
  (1930), 376--385.

\bibitem[Fer27]{Fermi27}
E.~Fermi, \emph{Un metodo statistico per la determinazione di alcune priorieta
  del atomo}, Rend. Accad. Nat. Lincei \textbf{6} (1927), 602--607.

\bibitem[GLM21]{GLM21}
D.~Gontier, S.~Lahbabi, and A. Maichine, \emph{Density functional theory for two-dimensional
homogeneous materials}, Commun. Math. Phys., \textbf{388(3)} (2021), 1475–1505.


\bibitem[HK64]{HK64}
P.~Hohenberg and W.~Kohn, \emph{Inhomogeneous electron gas}, Phys. Rev.
  \textbf{136} (1964), B864--B871.

\bibitem[Lie83]{Lieb83}
E.H. Lieb, \emph{Density functionals for coulomb systems}, Int. J. Quantum
  Chem. \textbf{24} (1983), 243--277.

\bibitem[LKBK01]{LKBK01}
G. Lu, N. Kioussis, V.~V. Bulatov, and E. Kaxiras,
  \emph{Dislocation core properties of aluminum: a first-principles study},
  Materials Science and Engineering: A \textbf{309--310} (2001), 142--147.

\bibitem[LS77]{LS77}
E.H. Lieb and B.~Simon, \emph{The thomas--fermi theory of atoms, molecules and
  solids}, Advances in Mathematics \textbf{23} (1977), 22--116.

\bibitem[Nab47]{nabarro47}
F.~R.~N. Nabarro, \emph{Dislocations in a simple cubic lattice}, Proc. Phys.
  Soc. \textbf{59} (1947), 256.

\bibitem[NO17]{NO17}
F.Q. Nazar and C.~Ortner, \emph{Locality of the {T}homas--{F}ermi--von
  {W}eizs\"acker equations}, Arch. Rational Mech. Anal.
  \textbf{224} (2017), 817--870.

\bibitem[Pei40]{peierls40}
R. Peierls, \emph{The size of a dislocation}, Proc. Phys. Soc. \textbf{52}
  (1940), 34.

\bibitem[Ric18]{Ricaud18}
J.~Ricaud, \emph{Symmetry Breaking in the Periodic
{T}homas--{F}ermi--{D}irac--von {W}eizs\"acker
Model}, Ann. Henri Poincaré
  \textbf{18} (2018), 3129--3177.


\bibitem[Sol16]{Solovej16}
J.P. Solovej, \emph{A new look at {T}homas--{F}ermi theory}, Molecular Physics
  \textbf{114} (2016), 1036--1040.

\bibitem[SP14]{SP14}
P.~Suryanarayana and D.~Phanish, \emph{Augmented lagrangian formulation of
  orbital-free density functional theory}, Journal of Computational Physics
  \textbf{275} (2014), 524--538.

\bibitem[Tho27]{Thomas27}
L.H. Thomas, \emph{The calculation of atomic fields}, Proc. Camb. Philos. Soc.
  \textbf{23} (1927), 542--548.

\bibitem[Vit92]{Vitek92}
V.~Vitek, \emph{Structure of dislocation cores in metallic materials and its
  impact on their plastic behaviour}, Progress in Materials Science \textbf{36}
  (1992), 1--27.

\bibitem[Wei35]{vw35}
C.F.v. Weizs\"acker, \emph{Zur theorie der kernmassen}, Z. Physik \textbf{96}
  (1935), 431--458.

\end{thebibliography}
\providecommand{\bysame}{\leavevmode\hbox to3em{\hrulefill}\thinspace}
\providecommand{\MR}{\relax\ifhmode\unskip\space\fi MR }
\providecommand{\MRhref}[2]{%
  \href{http://www.ams.org/mathscinet-getitem?mr=#1}{#2}
}
\providecommand{\href}[2]{#2}

\end{document}